\documentclass[12pt,english]{article}
\usepackage[notref, notcite]{showkeys}
\usepackage{geometry}
\usepackage{longtable}
\usepackage{booktabs}
\usepackage{multirow}
\usepackage{setspace}
\usepackage[authoryear,round,colon, sort&compress]{natbib}
\bibpunct{(}{)}{,}{autheryear}{,}{,}
\usepackage{rotating}

\usepackage[english]{babel}
\usepackage{amsmath, amsthm, amssymb, amscd, amsfonts, bm, bbm, mathrsfs}   % math packages
\usepackage{graphicx}
\usepackage{enumerate}
\usepackage[shortlabels]{enumitem}
\usepackage{url} % not crucial - just used below for the URL 
\usepackage[colorlinks]{hyperref}
\usepackage{xr-hyper}
\usepackage[capitalize,nameinlink,noabbrev]{cleveref} % to emulate \autoref style
\usepackage[algo2e,linesnumbered,lined,ruled]{algorithm2e}

\usepackage[bottom]{footmisc}
\usepackage[dvipsnames]{xcolor}
\hypersetup{
	breaklinks,
	linkcolor=blue,
	urlcolor=blue,
	anchorcolor=blue,
	citecolor=blue
}

\makeatletter

\theoremstyle{plain}

\@ifundefined{date}{}{\date{}}
\usepackage{babel}
\allowdisplaybreaks

\usepackage{tikz}
\usetikzlibrary{fit,positioning}

\newcommand*{\indep}{%
  \mathbin{%
    \mathpalette{\@indep}{}%
  }%
}
\newcommand*{\nindep}{%
  \mathbin{%                   % The final symbol is a binary math operator
    \mathpalette{\@indep}{\not}% \mathpalette helps for the adaptation
  }%
}
\newcommand*{\@indep}[2]{%
  \sbox0{$#1\perp\m@th$}%        box 0 contains \perp symbol
  \sbox2{$#1=$}%                 box 2 for the height of =
  \sbox4{$#1\vcenter{}$}%        box 4 for the height of the math axis
  \rlap{\copy0}%                 first \perp
  \dimen@=\dimexpr\ht2-\ht4-.2pt\relax
  \kern\dimen@
  {#2}%
  \kern\dimen@
  \copy0 %                       second \perp
}

\newtheorem{thm}{Theorem}

\newtheorem{lem}{Lemma}
\newtheorem{prop}{Proposition}

\theoremstyle{definition}
\newtheorem{rem}{Remark}
\newtheorem{defn}{Definition}
\newtheorem{eg}{Example}
\newcommand{\argmin}{\operatornamewithlimits{argmin}} %argmin
\makeatletter
\newcommand*{\rom}[1]{\expandafter\@slowromancap\romannumeral #1@}
\newcommand{\HUGE}{\@setfontsize\Huge{40}{50}}   
\makeatother

\makeatletter
\newcommand{\labelcustom}[2]{%
	\protected@write \@auxout {}{\string \newlabel {#1}{{#2}{\thepage}{#2}{#1}{}} }%
	\hypertarget{#1}{#2}
}
\makeatother

\makeatletter
\newcommand{\labeltext}[3][]{%
	\@bsphack%
	\csname phantomsection\endcsname% in case hyperref is used
	\def\tst{#1}%
	\def\refmarkup{}%
	\ifx\tst\empty\def\@currentlabel{\refmarkup{#2}}{\label{#3}}%
	\else\def\@currentlabel{\refmarkup{#1}}{\label{#3}}\fi%
	\@esphack%
	\labelmarkup{#2}% visible printed text.
}
\makeatother

\makeatletter
\newcommand{\bianca}{\renewcommand\NAT@open{[}\renewcommand\NAT@close{]}}
\makeatother

\newcommand{\iid}{\overset{\mathsf{iid}}{\sim}} % follows iid
\newcommand{\ed}{\overset{\mathsf{d}}{=}} % the same in distribution

\newcommand{\pr}{\mathsf{P}}
\newcommand{\eo}{\mathsf{E}}

\newcommand{\var}{\mathsf{var}}

\newcommand{\nd}{\mathsf{N}}

\newcommand{\ap}{\alpha} %alpha
\newcommand{\g}{\gamma} %gamma
\newcommand{\ga}{\Gamma} %Gamma
\newcommand{\dt}{\delta} % small delta
\newcommand{\Dt}{\Delta} % BIG Delta
\newcommand{\e}{\varepsilon} % epsilon
\newcommand{\ka}{\kappa} % kappa
\newcommand{\oa}{\omega} % omega
\newcommand{\s}{\sigma} % sigma
\newcommand{\sa}{\Sigma} % Sigma
\newcommand{\ld}{\lambda} % small lambda
\newcommand{\G}{\mathbb{G}} % G
\newcommand{\HH}{\mathbb{H}} % H
\newcommand{\I}{\mathbb{I}} % I
\newcommand{\M}{\mathbb{M}} % M
\newcommand{\N}{\mathbb{N}} % N
\newcommand{\R}{\mathbb{R}} % R
\newcommand{\SSS}{\mathbb{S}} %S 
\newcommand{\W}{\mathbb{W}} % W
\newcommand{\dd}{\mathcal{D}}	% D
\newcommand{\ee}{\mathcal{E}}	% E
\newcommand{\ii}{\mathcal{I}}	% I
\newcommand{\jj}{\mathcal{J}}	% J
\newcommand{\kll}{\mathcal{L}}	% L
\newcommand{\pp}{\mathcal{P}}	% P
\newcommand{\ttt}{\mathcal{T}}	% T
\newcommand{\vv}{\mathcal{V}}	% V
\let\oldnl\nl
\newcommand{\nlnonumber}{\renewcommand{\nl}{\let\nl\oldnl}}

\renewcommand{\eqref}[1]{(\ref{#1})}

\renewcommand{\sectionautorefname}{Section}
\renewcommand{\subsectionautorefname}{Section}
\renewcommand{\subsubsectionautorefname}{Section}
\renewcommand{\algorithmautorefname}{Algorithm}

\newcommand{\opnorm}[1]{
	{\vert\kern-0.25ex\vert\kern-0.25ex\vert #1 \vert\kern-0.25ex\vert\kern-0.25ex\vert}
}
\newcommand{\opnormbig}[1]{
	{\left\vert\kern-0.25ex\left\vert\kern-0.25ex\left\vert #1 \right\vert\kern-0.25ex\right\vert\kern-0.25ex\right\vert}
}

\newcommand{\opip}[1]{
	{\langle\kern-0.25ex\langle #1 \rangle\kern-0.25ex\rangle}
}
\newcommand{\opipbig}[1]{
	{\left\langle\kern-0.25ex\left\langle #1 \right\rangle\kern-0.25ex\right\rangle}
}

\begin{document}

	\newcommand*{\Appendixautorefname}{Appendix}
	\renewcommand{\sectionautorefname}{Section}
	\renewcommand{\subsectionautorefname}{Section}
	\renewcommand{\subsubsectionautorefname}{Section}
	\renewcommand{\algorithmautorefname}{Algorithm}

\title{
%	Effective and flexible 
	Depth-based inference 
	for functional parameters
}
\author{
	Hyemin Yeon\thanks{
		Department of Statistics, George Mason University, Fairfax, VA, USA % \\ Email: hyeon2@gmu.edu.
	} 
%		~and
%		Sara Lopez-Pintado\thanks{
%		Department of Health Sciences, Northeastern University, Boston, MA, 02115, USA, Email: s.lopez-pintado@northeastern.edu.
%	}
%	and Daniel John Nordman\thanks{
%		Department of Statistics, Iowa State University
%	} 
}

\newpage

\maketitle

%\begin{center}
%	\date{\today}
%\end{center}

\begin{abstract}
	
	\noindent
	For hypothesis testing of functional parameters, given a functional statistic $T_n$ and a functional depth $D$ with respect to the sampling distribution $P_n$ of $T_n$, we propose the depth value $D(T_n;P_n)$ as a test statistic, which we call a depth statistic. 
	In practice, its sampling distribution is approximated by a resampling method such as the bootstrap. 
	We establish validity of this approach under mild conditions on the sampling distribution, the resampling distribution, and the functional depth, and we verify these conditions for several representative functional depths. 
	The framework applies to inference for functional parameters, including settings in which both the data and the parameter of interest may be infinite-dimensional. 
	While maintaining accurate size, the proposed depth statistic can provide higher power against alternatives with complex functional patterns, especially subtle shape differences. 
	We study its performance in two-sample functional mean testing and mean response inference in function-on-function regression, 
	and illustrate the method with Canadian weather and German electricity price data.
	
%	For hypothesis testing of functional parameters, 
%	given a functional statistic $T_n$ and a functional depth $D$ with respect to the distribution $P_n$ of $T_n$,
%	we propose the depth value $DT_n \equiv D(T_n;P_n)$ as a test statistic,
%	which we refer to as a \emph{depth statistic}. 
%	In practice, its sampling distribution is approximated by a resampling method such as bootstrap.
%	While achieving accurate sizes, a test based on the proposed depth statistic produces stronger power,
%	as it remains sensitive even to subtle variations arising from complex functional patterns in the alternatives. 
%	Moreover, it is broadly applicable to a broad range of inference problems for functional parameters, including two-sample tests, analysis of variance, regression, etc.
%	We provide its theoretical guarantee under mild assumptions
%	along with examples of bootstrap methods and functional depths that satisfy these conditions.
%	Its effectiveness is thoroughly investigated through numerical studies under two popular frameworks: 
%	(i) two-sample functional mean tests and (ii) mean response inference for function-on-function regression.
%	The proposed depth statistic is illustrated with two data examples: Canadian weather and German electricity prices datasets. 

\medskip \noindent
\textit{Keywords and phrases:} 
Data depth,
Functional data analysis,
Infinite-dimensional parameters; 
Shape alternatives.

\end{abstract}

%\newpage
%\tableofcontents

\newpage

\section{Introduction}

\subsection{Background}

%In functional data analysis, 
%statistical inference for functional parameters is challenging due to infinite dimensionality of the underlying function space. 
%Challenges may arise because both data and parameters are infinite-dimensional.
In functional data analysis (FDA), statistical inference is challenging because \textit{both} the observed data and the parameters of interest reside in infinite-dimensional function spaces.
Nevertheless, numerous meaningful advances have been made.
These include 
(a) confidence bands for functional parameters by \cite{Degras11, CR18, LR23},
(b) two-sample tests for 
equality of
(i) means by \cite{BIW11, HKR13, BHK09, PSV18, DKA20},
(ii) covariance operators by \cite{PKM10, KP12, FSHK13}, and
(iii) eigenvalues or eigenfunctions by \cite{BHK09, ZS15, ADR23},
(c) functional ANOVA for 
(i) several means by
\cite{CFF04, CF10, ZL14, HR15, GS15, PS16, ZCWZ19, LLM23} and
(ii) multiple covariance operators by \cite{PS16, BRS18, GZZ19, MPZ24},
(d) tests of equality between distributions by \cite{HV07, WD22, HHK22, CLWW25},
(e) significance tests for contrast functions in 
(f) function-on-scalar regression models by \cite{fara97, SF04, ZC07, Zhang11} and
(g) functional concurrent regression models by \cite{GM22, WPCL18},
(h) inference for mean response or slope operator in function-on-function regression models by \cite{KMSZ08, Yeon26, CM13, DT24}.
Classical summary statistics in this literature often take the form of
quadratic structures incorporating eigendecay (e.g., truncated Hotelling's $T^2$), 
integration (e.g., squared $L^2$ norm), 
supremum (e.g., supremum of pointwise F statistics),
or their variants,
where their performance can be enhanced by resampling methods such as bootstrap.

Complementary to inference for functional parameters, 
%\tred{There is a lack of some introduction to functional depths. A few general words for non-specialists would be helpful.}
considerable attention has also been devoted to the study of functional depths. 
Originally introduced in multivariate statistics \citep{Tukey75, Liu90, Chau96, Mosler02, Zuo03}, 
a depth function assigns a nonnegative value to each point in the space of the observed random elements, and this value indicates how central or outlying a data point is relative to a specified distribution.
The notion of data depth has been extended to functional data, 
leading to proposals such as 
integrated \citep{FM01}, 
kernel (or $h$-mode) \citep{CCF06}, 
(modified) band \citep{LR09}, 
half-region \citep{LR11}, 
infimal \citep{Mosler13}, 
spatial \citep{CC14b}, 
extremal \citep{NN16}, 
regularized Mahalanobis \citep{BBC20}, 
$L^2$ \citep{RC24},
spherical \citep{MN25},
and regularized halfspace \citep{YDL25RHD} depths;
see \autoref{secFdepth} for a brief overview of some selective functional depths.
These constructions provide rankings of complex functional observations that capture the geometry of the underlying distribution, 
beyond simple balls or ellipsoids determined solely by norms or quadratic forms. 
While functional depths have enabled a wide range of statistical applications 
such as 
rank-based tests \citep{GKN24}, 
classification \citep{HRS17}, and 
outlier detection \citep{NGH17},
deriving their theoretical properties is more challenging than in the multivariate setting, as standard finite-dimensional techniques do not readily extend to infinite-dimensional spaces.

Since there have been thousands of contributions in the vast literature on both (i) inference for functional parameters and (ii) functional depths,
we do not attempt an exhaustive review here.

\subsection{Depth statistic for functional parameters}

In this work, we propose using the depth evaluated at a statistic, 
as a test statistic for general statistical inference problems for functional parameters.
This approach will be able to bridge the aforementioned two active communities: statistical inference in FDA and functional depth.
Concretely, given (i) a test statistic $T_n$ to test a hypothesis $H_0:\theta = \theta_0$ about a functional parameter $\theta$ and (ii) a functional depth $D$ at hand, 
we define the depth statistic as $DT_n \equiv D(T_n; P_n)$,  
the depth value evaluated at $T_n$ with respect to its sampling distribution $P_n$,
where small values of $DT_n$ support evidence against $H_0$.
Although the sampling distribution of the depth statistic $DT_n$ cannot be explicitly provided,
we employ resampling methods to approximate them, 
such as well-established bootstrap methods for functional data \cite[e.g.,][]{BHK09, PS16, Yeon26}.
The proposed depth statistic then leads to valid inference tools for functional parameters,
as long as 
(i) the resampled distribution $\widehat{P}_n$ consistently estimate the original sampling distribution $P_n$ with limiting distribution $P$, 
(ii) the functional depth possesses continuity in either element or probability distribution, and 
(iii) the depth value at the limiting random element with respect to $P$ is continuous.

The depth statistic exploits the entire ordering structure of the sampling distribution for hypothesis testing.
It offers not only accurate sizes but also greater power than classical summary statistics,
which often fail to detect complex alternatives
whose signals hindered by complexity in functional data.
To illustrate, we consider a realization of 50 curves under the two-sample functional mean framework (\autoref{egFANOVA}(a)), 
with 25 curves per group (cf.~\autoref{ssec_5_1}).
\autoref{fig_2test_eg} displays 
the two-sample test statistic $T_{\mathrm{two},\bm{n}} = \bar{X}_1 - \bar{X}_2$ in green, 
the true mean difference $\theta = \mu_1-\mu_2$ in red,
and 1,000 bootstrap statistics $T_{\mathrm{two},\bm{n}}^* = \bar{X}_1^* - \bar{X}_2^*$ (cf.~\autoref{ssec_2test}) in gray.
Although the problem is challenging because the group means differ only in shape,
the depth statistics effectively capture subtle functional alternatives,
with the regularized halfspace depth and kernel depth (cf.~\autoref{secFdepth}) yielding p-values below the significance level 0.05
(0.000 for the former with quantile level 0.1 and 0.009 for the latter with quantile level 0.01).
In contrast, classical statistics often struggle to detect complex signals in functional alternatives, 
as illustrated by the commonly used $L^2$ and supremum norm–based tests. 
In \autoref{fig_2test_eg}, the 95\% confidence band constructed from the supremum norm (shown in blue), with critical values calibrated via bootstrap statistics, still contains zero even though the true mean difference deviates from zero in shape, yielding a p-value of 0.098.
This indicates that the supremum statistic fails to capture the functional characteristic in shape. 
Similarly, the $L^2$ statistic produces a p-value above the 0.05 significance level (0.081). 
Further details on the performance of the proposed depth statistics are provided in \autoref{sec5}.

\begin{figure}[h!]
	\centering
	\includegraphics[width=0.7\linewidth]{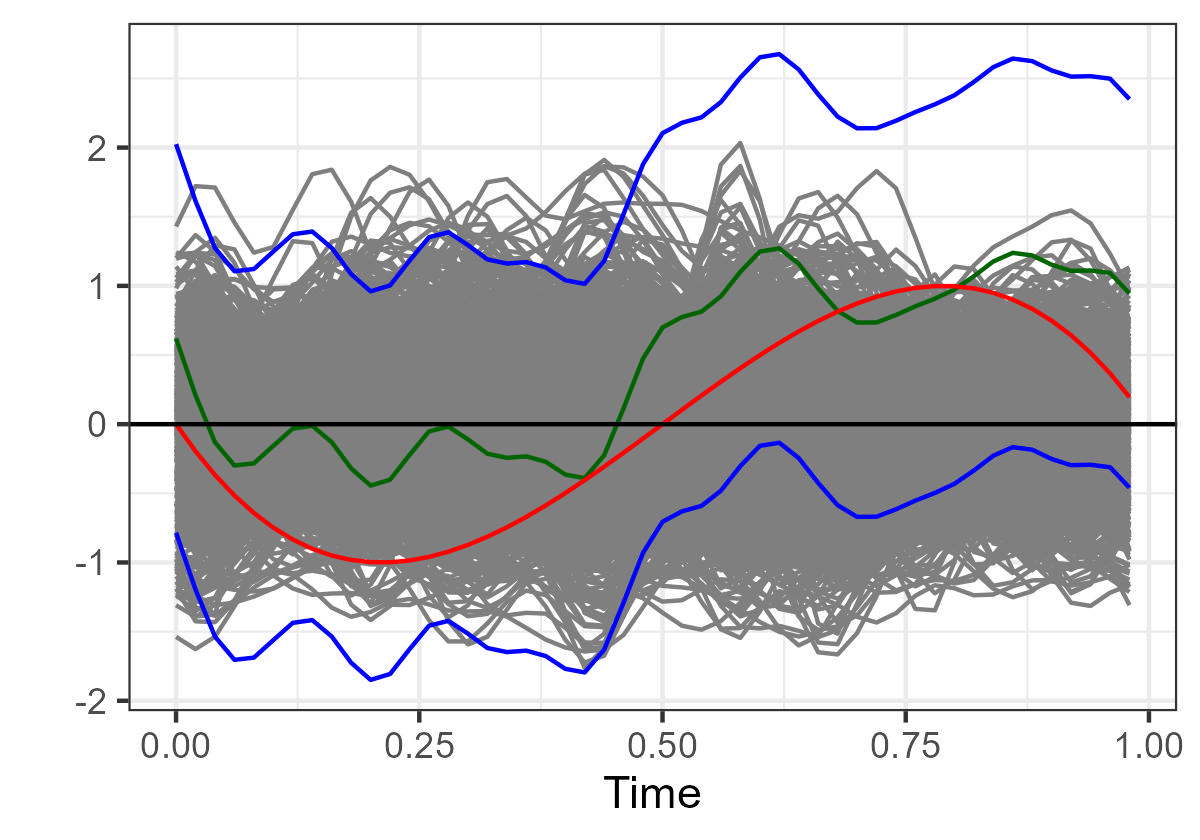}
	\caption{
		The test statistic $T_{\mathrm{two},\bm{n}} = \bar{X}_1-\bar{X}_2$ (green), 
		bootstrap statistics $T_{\mathrm{two},\bm{n}}^* = \bar{X}_1^* - \bar{X}_2^*$ (grey),
		the true mean difference $\theta = \mu_1-\mu_2$ (red, cubic alternative),
		and the 95\% confidence band constructed from the supremum norm (blue)
		The black horizontal line at zero stands for the true mean difference in the null hypothesis $H_0:\mu_1=\mu_2$.  
		The $n_k=25$ functional observations in each group $k$ are generated as described in \autoref{ssec_5_1} under equal covariance scenario when the functional principal components are of NN-type.
		The details about the two-sample bootstrap procedure can be found in \autoref{ssec_2test}. 
	}
	\label{fig_2test_eg}
\end{figure}

Beyond its effectiveness, 
the proposed depth statistic is broadly applicable in two senses. 
First, it can be used for a wide range of inference problems in FDA
whenever the functional depth is well estimated and the resampling distribution consistently approximates the sampling distribution.
In \autoref{secInferFDA}, we present several functional inference frameworks, including those discussed above, and highlight two fundamental applications: 
(i) two-sample functional mean testing and (ii) mean response inference in function-on-function regression. 
Second, the proposed procedure is compatible with many popular functional depths. 
In particular, in \autoref{secFdepth}, we consider kernel depth, regularized halfspace depth, integrated depth, and infimal depth, and verify under mild assumptions that they satisfy our technical conditions.
%The corresponding theoretical foundations and numerical results are presented in Sections~\ref{secInferFDA}--\ref{secFdepth} and~\ref{sec5}, respectively.

In developing our functional depth statistic \(DT_n\), we noted related ideas for multivariate data in \cite{LS97,YS97}. 
Our setting and goals, however, are substantially different from those of these earlier works.
Both papers are primarily concerned with finite-dimensional settings. In particular, \citet[Section~5]{LS97} considers an infinite-dimensional parameter only in the special case of cumulative distribution functions (CDFs), 
while the data remain finite-dimensional. 
There, the depth is defined through a distance between CDFs relative to an additional reference distribution, 
and the validity in their Theorem~5.1 relies on asymptotic properties of empirical CDFs. 
Consequently, that argument does not directly extend to general functional depths or to inference problems 
in which the data are infinite-dimensional as well. 
This distinction is especially important in FDA, where departures in a functional parameter are reflected in the data through global functional features such as shape, rather than through a few finite-dimensional coordinates (cf.~\autoref{fig_2test_eg}), 
so the practical benefits of depth-based procedures can be more pronounced than in finite-dimensional problems. 
To the best of our knowledge, depth-based inference has not been studied in the setting 
where \textit{both} the data and the parameter are infinite-dimensional, 
and no comprehensive power study appears to have been conducted in this setting.

Our theoretical contribution is also substantial. 
Unlike \cite{LS97,YS97}, our depth statistic is defined at the population level, and the sufficient conditions underlying our main result are weaker than those in \cite{LS97}. 
More importantly, for functional depths, these conditions are genuinely delicate rather than routine: 
in infinite-dimensional spaces, properties such as consistency over an entire space and continuity with respect to the underlying probability measure are no longer automatic \cite[cf.][]{WN25}. 
We show that all depths considered in this paper satisfy these conditions under mild assumptions. 
In particular, for kernel depth we establish new continuity results, including uniform continuity over the whole space (\autoref{ssecFdepthKD}), 
while for regularized halfspace depth, we derive continuity in measure through a new technique based on continuity in regularization and properties of the weak topology in functional analysis (\autoref{ssecFdepthRHD}).

\subsection{Organization of the paper}

%The structure of the paper is organized as follows. 
%After summarizing preliminary notation,
%in \autoref{secDepthInfer}, the formal definition of the proposed functional depth statistic is presented along with its theoretical validity.
%Sections~\ref{secInferFDA} and~\ref{secFdepth} then provide various examples of inference problems and depth in FDA,
%where our framework is applicable.
%Importantly, we found some cases that satisfies our technical conditions.
%Namely, we develop detailed procedures for functional depth-based inference in the contexts of two-sample testing and function-on-function regression in \autoref{secInferFDA};
%and we establish continuity properties of all functional depths under consideration in \autoref{secFdepth}. 
%\autoref{sec5} reports numerical power analyses that demonstrate the superiority of the proposed depth statistic, 
%while \autoref{sec6} provides empirical illustrations using real data. 

The paper is organized as follows. 
After introducing notation, \autoref{secDepthInfer} defines the proposed functional depth statistic and establishes its validity.
Sections\ref{secInferFDA} and\ref{secFdepth} present inference problems and functional depths in FDA to which the framework applies. 
In particular, \autoref{secInferFDA} introduces several inferential problems in FDA, including functional two-sample mean testing and function-on-function regression, 
while \autoref{secFdepth} verifies the required continuity properties for all depths considered in the paper. 
Numerical power studies and real-data illustrations are given in \autoref{sec5} and \autoref{sec6}, respectively.
All proofs are deferred to the supplementary material, 
and the R implementations of the proposed methods are publicly available through the author’s GitHub repository.\footnote{https://github.com/luckyhm1928}

\subsection{Notation}

We fix the notation, which is used throughout the paper.
Our methodology is founded on a Hilbert space framework.
We use $\HH, \HH', \HH''$ to represent infinite-dimensional separable Hilbert spaces;
this distinction will be needed only when data and parameters spaces are different. 
As an example, we mainly work on the space $L^2([0,1])$ of all square-integrable real-valued functions on $[0,1]$,
where the inner product there is defined as $\langle f, g \rangle \equiv \int_0^1 f(t)g(t)dt$ for $f,g \in L^2([0,1])$ and induces norm $\|\cdot\|$. 
For univariate functions $f_1,\dots,f_p$ on some common domain $E$, $\bm{f} = (f_1, \dots, f_p)^\top:E \to \R^p$ is the multivariate function defined by $\bm{f}(t) = (f_1(t), \dots, f_p(t))^\top$ for $t \in E$.
In particular, if $\{f_1, \dots, f_p\} \subseteq \HH = L^2([0,1])$, then $\bm{f} = (f_1, \dots, f_p)^\top$ resides in another separable Hilbert space $\HH' = \{L^2([0,1])\}^p$. 
For $x$ in a Hilbert space $\HH$, $x^{\otimes 2} = x \otimes x$ is a shorthand expression to denote the squared tensor product operator of $x$ on $\HH$, defined by $x^{\otimes 2}(f) = \langle x, f \rangle x$ for $f \in \HH$.

For a metric space $\M$, $\pp(\M)$ stands for the set of all probability distributions on $\M$,
and let $\pi$ denote the Prohorov metric on the space $\pp(\M)$ of probability measures on $\M$ \cite[cf.][Section~5]{bill99};
it is useful to recall that $\pi$ characterizes the weak convergence in $\pp(\M)$.
We reserve $P$ to denote the distributions of the statistics or limiting random element,
while $Q$ is used for a generic notation for probability distributions.
For $d \in \N$, the set of all continuous real-valued functions from $[0,1]$ to $\R^d$ is denoted by $C^d([0,1])$ equipped with supremum norm $\|\cdot\|_{\sup}$ defined as $\|f\|_{\sup} \equiv \sup_{u \in [0,1]} \|f(u)\|_{\R^d}$ for $f \in C^d([0,1])$;
we write $C([0,1]) = C^d([0,1])$ when $d=1$.
Also, for $Q \in \pp( C^d([0,1]) )$, a random element $Z \sim Q$, and each $u \in [0,1]$,
we write $Q_u$ for denoting the probability distribution of $Z(u)$ on $\R^d$.

\section{Depth-based inference for functional parameters}

\label{secDepthInfer}

Here, we describe the depth statistic in a general Hilbert space framework.
We suppose data observations $\dd_{\bm{n}}$ including functional data valued in $\HH$ are given for statistical inference for functional parameter $\theta \in \HH'$.
Here, $\bm{n}$ stands for the collection of all sample sizes, 
as we may observe multiple samples;
$\bm{n} = n$ for one-sample cases.
With estimator $\hat{\theta}$ for $\theta$, we may construct a functional statistic $T_{\bm{n}} \equiv A_{\bm{n}}(\hat{\theta}-\theta_0)$ for testing $H_0:\theta = \theta_0$ for some $\theta_0 \in \HH'$.
%which might be inspired by the central limit theorem for $\hat{\theta}$ given by $A_{\bm{n}} (\hat{\theta} - \theta) \xrightarrow{\mathsf{d}} P$ for some distribution $P$ on $\HH$ (if exists).
Here, $A_{\bm{n}}:\HH'\to\HH''$ is a bounded linear operator from $\HH'$ to a possibly different separable Hilbert space $\HH''$ 
with minimum eigenvalue $\ld_{A,\min,\bm{n}}$ being positive but diverging.
%\tred{?? mostly, $\HH''=\HH'$}
This operator represents a normalizing or scaling term.
While there are various cases where $\HH,\HH',\HH''$ are different, sample sizes are multiple, and scaling could be an operator (cf.~\autoref{secInferFDA}),
we describe our proposal assuming $\HH = \HH' = \HH''$, $\bm{n} = n$, and scalar scaling,
since the argument will not lose generality.

Let $P_n$ denote the (true) sampling distribution $P_n$ of $T_n$ under $H_0$,
i.e., $P_n(A) = \pr(T_n \in A|H_0)$ for a Borel set $A$ in $\HH$. 
Consider a depth function $D(\cdot;  P_n)$ on $\HH$ with respect to $P_n$ of $T_n$. 
Then, to test $H_0:\theta = \theta_0$, we propose using the \emph{depth statistic} defined by 
\begin{align} \label{eq_depth_stat}
	DT_n \equiv D(T_n; P_n).
\end{align}
We note that small values of $DT_n$ give evidence against $H_0$.

While the depth statistic in \eqref{eq_depth_stat} is flexible,
it requires finding the sampling distribution $P_n$ and computing the distribution $F_n$ of the depth statistic, which may not have explicit forms. 
Here, $F_n$ stands for the CDF of $DT_n$ under $H_0$, i.e., $F_n(u) = \pr(DT_n \leq u|H_0)$ for $u \in \R$.
For this, we estimate $F_n$ by employing a resampling method.
Suppose we have a consistent resampling method at hand
so that we construct resampling statistic
$T_n^* \equiv A_n(\hat{\theta}^* - \hat{\theta})$
with its (resampling) distribution $ \widehat{P}_n$ under $H_0$.
%Denoting $\widehat{D}_n \equiv D(\cdot; \widehat{P}_n)$, 
The resampling depth statistic is given by
\begin{align} \label{eq_depth_stat_boots}
	DT_n^* \equiv D(T_n^*; \widehat{P}_n)
\end{align}
with CDF $\widehat{F}_n \equiv \pr^*( DT_n^* \leq \cdot|H_0)$ under $H_0$.

We will show that, given the consistency of the sampling and resampling distributions and some continuity properties of the depth, 
the resampling depth statistic distribution $\widehat{F}_n$ consistently estimates its original version $F_n$.
Before stating the theorem, we list the technical assumptions. 
\begin{enumerate}[(C1)]
	\item \label{condConvDistnData}
	There exists $P \in \pp(\HH)$ such that $\pi(P_n, P)\to0$ as $n\to\infty$;
	
	\item \label{condConvDistnBTS} $\pi(\widehat{P}_n, P_n) \xrightarrow{\pr} 0$ as $n\to\infty$;
	
	\item \label{condContPoint}
	the depth function $D(\cdot;P)$ is continuous 
%	almost surely with respect to $P$
%	in the sense that
%	$\pr(D(\cdot;P) \text{ is discontinuous at }X) = 0$;
	
%	for a sequence $\{x_n\}$ in $\HH$ and $x \in \HH$ such that $x_n\to x \in \HH$, 
%	$\{D(x_n;Q)\}$ converges to $D(x;Q)$ uniformly over $Q \in \pp(\HH)$, 
%	i.e.,
%	$\sup_{Q \in \pp(\HH)} |D(x_n;Q) - D(x;Q)| \to 0$;
	
	\item \label{condContProb}
	for a compact subset $G \subseteq \HH$ and
	a sequence $\{Q_n\}$ in $\pp(\HH)$ such that $\pi(Q_n,P) \to 0$,
	the sequence $\{D(x;Q_n)\}$ converges to $D(x;P)$ uniformly over $x \in G$,
	i.e., 
	$\sup_{x \in G} |D(x;Q_n)-D(x;P)| \to 0$; and

%	for a sequence $\{Q_n\}$ in $\pp(\HH)$ and $Q \in \HH$ such that $\pi(Q_n,Q) \to 0$,
%	$\{D(x;Q_n)\}$ converges to $D(x;Q)$ uniformly over $x \in \HH$,
%	i.e., 
%	$\sup_{x \in \HH} |D(x;Q_n)-D(x;Q)| \to 0$; and
	
	\item \label{condLimitCont}
	for $T \sim P$, 
	the CDF of the random variable $D(T;P)$ is continuous.
\end{enumerate}

Condition~\ref{condConvDistnData} means that the sampling distribution $P_n$ weakly converges to some limiting distribution $P \in \pp(\HH)$,
while in Condition~\ref{condConvDistnBTS}, the resampling distribution $\widehat{P}_n$ consistently estimates $P_n$,
which are common in bootstrap literature.
%implying its weak convergence to $P$. 
Condition~\ref{condContPoint} means a continuiuty of $D$ in $x \in \HH$ only at the limit distribution $P$.
%In Conditions~\ref{condContPoint}--\ref{condLimitCont}, we require the depth function to be continuous in either $x \in \HH$ and $P \in \pp(\HH)$.
The continuity condition \ref{condContProb} in measure is required to hold only over a a compact set and very mild; 
this is especially important for infinite-dimensional data
because the uniform convergence over the entire function space may not be feasible,
as discussed \cite{WN25}.
This also distinguishes our work with \cite{LS97},
which requires the uniform continuity over the entire space.
We remark that Condition~\ref{condLimitCont}, though mild, may fail for degenerate functional depths \citep{GN17}. 
To assure practical relevance of these technical conditions, 
we will provide concrete examples that satisfy Conditions~\ref{condConvDistnData}--\ref{condConvDistnBTS} in \autoref{secInferFDA} and Conditions~\ref{condContPoint}--\ref{condLimitCont} in \autoref{secFdepth}, respectively.
Our main theoretical result is given below, 
while the proof is deferred to \autoref{AppProofMain} of the supplement.

\begin{thm} \label{thmDepthStat}
	Suppose that Conditions~\ref{condConvDistnData}--\ref{condLimitCont} hold. 
	Then, the distribution $\widehat{F}_n$ of the resampling depth statistic $DT_n^*$ consistently estimates the distribution $F_n$ of the original depth statistic $DT_n$ in the sense that 
	\begin{align*}
		\sup_{u \in \R} |\widehat{F}_n(u) - F_n(u)| \xrightarrow{\pr} 0.
	\end{align*}
	
\end{thm}

In practice, like other resampling-based methods, we approximate the resampling depth statistic distribution $\widehat{F}_n$ using Monte Carlo iterations.
The procedure  is summarized as follows and it can be applied to a general situation provided that the resampling consistency is guaranteed.
For $b=1,\dots,B$, we draw resampled data $\dd_{b,n}$ to compute the resampling statistic $T_{b,n}^*$.
Then, $\widehat{F}_n$ is approximated by
\begin{align} \label{eq_depth_boots_distn}
	\widetilde{F}_n (u)
	= \widetilde{F}_{n,B} (u)
	\equiv B^{-1} \sum_{b=1}^B \I(D(T_{b,n}^*;\widetilde{P}_n) \leq u), \quad u \in \R,
\end{align}
where $\widetilde{P}_n = \widetilde{P}_{n,B}$ denotes the empirical distribution of the resampled statistics $\ttt_{B,n} \equiv \{T_{b,n}^*\}_{b=1}^B$.

\medskip
\noindent
\textbf{Depth p-value} \ \
A natural application of the proposed depth statistic is the \emph{depth p-value}  defined by $\pi_n \equiv F_n(DT_n)$.
With $\widehat{DT}_n \equiv D(T_n; \widehat{P}_n)$, the resampling depth p-value $\hat{\pi}_n$ is then defined by
\begin{align} \label{eq_depth_pval_boots}
	\hat{\pi}_n \equiv \widehat{F}_n(\widehat{DT}_n).
\end{align}
This resampling depth p-value is indeed valid
because it asymptotically follows the standard uniform distribution $\mathsf{U}(0,1)$ as asserted in \autoref{thmDepthPval}.
The result is a direct consequence of \autoref{thmDepthStat},
whose proof is given in \autoref{AppProofMain} of the supplement. 
\begin{thm} \label{thmDepthPval}
	Suppose that Conditions~\ref{condConvDistnData}--\ref{condLimitCont} hold. 
	Then, the resampling depth p-value $\hat{\pi}_n$ in \eqref{eq_depth_pval_boots} converges in distribution to the standard uniform distribution, i.e., 
	$\hat{\pi}_n \xrightarrow{\mathsf{d}}  \mathsf{U}(0,1)$.
\end{thm}

Using the approximate distribution $\widetilde{F}_n$ in \eqref{eq_depth_boots_distn},
the estimated depth p-value $\hat{\pi}_n$ by a resampling method is approximated via Monte Carlo approximation by 
\begin{align} \label{eq_depth_pval_boots_approx}
	\tilde{\pi}_n 
	= \tilde{\pi}_{\bm{n},B}  \equiv \widetilde{F}_n( \widetilde{DT}_n ),
\end{align}
where $\widetilde{DT}_n \equiv D(T_n; \widetilde{P}_n)$.
We use the approximate p-value in \eqref{eq_depth_pval_boots_approx}
for testing the null hypothesis $H_0:\theta = \theta_0$.
Namely, if the p-value $\tilde{\pi}_n$ is no bigger than the significance level, then we reject $H_0$. 

\begin{rem}
	One could use the proposed depth statistic to construct a confidence region for infinite-dimensional parameter $\theta$,
	e.g., by inverting the test for $H_0:\theta = \theta_0$ \citep[Section~9.2.1]{CB02}. 
	The subsequent confidence region may be given by
	\begin{align*}
		\widehat{C}_n \equiv \{f \in \HH': D(A_n(\hat{\theta}-f);\widehat{P}_n) \geq \widehat{F}_n^{-1}(\ap)\},
	\end{align*}
	where its $1-\ap$ asymptotic coverage can be derived from \autoref{thmDepthStat} under Conditions~\ref{condConvDistnData}--\ref{condLimitCont} \citep[cf.][Lemma~F.7]{YK26}.
	Nevertheless, we do not pursue these tools because using them for testing $H_0:\theta = \theta_0$ is equivalent to testing $H_0$ with the depth p-value in \eqref{eq_depth_pval_boots}.
\end{rem}

\begin{rem}
	We consider simple null hypotheses in the present paper.
	Extensions to composite null hypotheses as in \citet[Section~4]{LS97}, including relevant hypotheses considered in \cite{DKA20}, are possible but beyond the scope of this work; see also the discussion in the conclusion.
%	\tred{
%	While not stated explicitly, the current method assumes that the null hypothesis is simple.
%	This can be seen in formula (1), where Pn is the distribution of the test statistics under
%	the null hypothesis. For a composite hypothesis, there is in general no unique distribution
%	under the null, and thus the present approach is not applicable. This severely limits the
%	present methods.}
\end{rem}

\section{Inference problems for functional data}
\label{secInferFDA}

As established in \autoref{thmDepthStat}, the proposed depth statistic demonstrates broad applicability. 
Among the numerous examples of functional parameter inference problems to which it may be applied, we elaborate on several representative cases below.

\begin{eg}[Functional ANOVA] \label{egFANOVA}
	Let $\{X_{ki}\}_{i=1}^{n_k}$, $k=1,\dots, K$, be $K \in \N$ independent random samples that take values in a Hilbert space $\HH$,
	where $\bm{n} = (n_1, \dots, n_K)$ denotes the collection of all sample sizes. 
	For each $k = 1, \dots, K$, we consider the following setup.
	Suppose that these functional observations have finite fourth moment, i.e., $\eo[\|X_{k1}\|^4]<\infty$, 
	and define their mean functions and covariance operators as $\mu_k \equiv \eo[X_{k1}] $ and $\ga_k \equiv \var[X_{ki}] = \eo[(X_{ki} - \mu_k)^{\otimes 2}]$.
	We write $\bar{X}_k \equiv n_k^{-1} \sum_{i=1}^{n_k} X_{ki}$ and $\widehat{\ga}_k \equiv n_k^{-1} \sum_{i=1}^{n_k}(X_{ki}-\bar{X}_k)^{\otimes 2}$ to denote the sample mean and sample covariance of group $k$, respectively.
	Furthermore, $\phi_{kj}$ and $\hat{\phi}_{kj}$ denotes the $j$-th eigenfunctions of $\ga_k$ and $\widehat{\ga}_k$, respectively.
	
	\begin{enumerate}[(a)]
		\item \textit{Two-sample tests.} \ \ 
		We first consider the case with $K=2$ and define
		$a_{\mathrm{two},\bm{n}} = (n_1^{-1}+n_2^{-1})^{-1/2}$. 
		The null hypotheses and corresponding functional statistics are summarized in \autoref{tb_eg_two} depending on the types of the tests. 
		These problems have been widely studied by
		\cite{HKR13, FHKS14, DKA20} for mean,
		\cite{PKM10, PADS14} for covariance, and
		\cite{BHK09, ZS15, ADR23} for eigenfunctions, 
		to name a few.
		\begin{table}[h!]
			\renewcommand{\arraystretch}{1.3}
			\centering
			\caption{
				The null hypotheses and test statistics for different functional two-sample tests.
				}
			\label{tb_eg_two}\medskip
			\begin{tabular}{c|cc} \hline
				Tests & Null hypotheses & Test statistics  \\ \hline
				Mean & $\mu_1 = \mu_2$ & $a_{\mathrm{two},\bm{n}}(\bar{X}_1 - \bar{X}_2)$  \\
				Covariance & $\ga_1 = \ga_2$ & $a_{\mathrm{two},\bm{n}}(\widehat{\ga}_1 - \widehat{\ga}_2)$ \\
				Eigenfunction & $\phi_{1j} = \phi_{2j}$ & $a_{\mathrm{two},\bm{n}}(\hat{\phi}_{1j} - \hat{\phi}_{2j})$ \\ \hline
			\end{tabular}
		\end{table}

		\item 
		\textit{Test for equality between several means.}  \ \
		For general $K$, we are interested in testing the equality between several means stated as $H_{K,\mathrm{mean},0}:\mu_1=\cdots=\mu_K$.
		We re-parametrize this problem with $$\bm{\theta} = (\mu_2-\mu_1, \dots, \mu_K-\mu_1)^\top$$
		so that $H_{K,\mathrm{mean},0}$ is equivalent to the statement that $\bm{\theta} = \bm{0}$. 
		Then, we consider the multivariate functional statistic
		$$T_{K,\mathrm{mean},\bm{n}} \equiv (a_{2\bm{n}}(\bar{X}_2 - \bar{X}_1), \dots, a_{K\bm{n}}(\bar{X}_K-\bar{X}_1))^\top$$ with
		$a_{k\bm{n}} = (n_k^{-1} + n_1^{-1})^{-1/2}$ for $k=2,\dots,K$.
		Other parametrizations are possible such as $\bm{\theta} = (\mu_1-\mu_K, \dots, \mu_{K-1}-\mu_K)^\top$.
		Some selective studies about this problem include \cite{CFF04, ZL14, GS15, PS16, ZCWZ19}.
		
	\end{enumerate}
	
	%	\begin{enumerate}[(a)]
		%		\item Two-sample mean test.
		%		The null hypothesis is $H_{\mathrm{two},\mathrm{mean},0}:\mu_1 = \mu_2$ and we can use $T_{\mathrm{two},\mathrm{mean},\bm{n}} = a_{\bm{n}}(\bar{X}_1 - \bar{X}_2)$ for our functional statistic.
		%		\tred{Popular depth functions can be used to construct the depth p-value}
		%		
		%		\item Two-sample covariance test.
		%		The null hypothesis is $H_{\mathrm{two},\mathrm{cov},0}:\ga_1 = \ga_2$ and we can use $T_{\mathrm{two},\mathrm{cov},\bm{n}} =a_{\bm{n}}(\widehat{\ga}_1 - \widehat{\ga}_2)$ for the operator statistic (valued in $\bb_2(\HH)$) to test $H_{\mathrm{two},\mathrm{cov},0}$. 
		%		\tred{Depth functions for bi-variate domain functions can be used to construct the depth p-value}
		%		
		%		\item Two-sample test for eigenfunctions.
		%		The null hypothesis is $H_{\mathrm{two},\mathrm{ef},j,0}:\phi_{1j} = \phi_{2j}$
		%		and we can use $T_{\mathrm{two},\mathrm{ef},j,\bm{n}} = a_{\bm{n}}(\hat{\phi}_{1j} - \hat{\phi}_{2j})$ for functional statistic to test $H_{\mathrm{two},\mathrm{ef},j,0}$.
		%		
		%	\end{enumerate}
\end{eg}

\begin{eg}[Functional regression] \label{egFunReg}
	We illustrate the wide applicability of the proposed depth statistic with several functional regression models with functional responses.
	Below, let $\{Y_i\}_{i=1}^n$ and $\{\e_i\}_{i=1}^n$ respectively denote functional responses and functional errors that take values in $\HH = L^2([0,1])$, with sample size $\bm{n}=n$,
	and $\ap \in \HH$ represent the intercept function,
	while the structures of regressors can be different depending on models.
	\begin{enumerate}[(a)]

		\item \textit{Functional concurrent regression}: 
		\begin{align*}
			Y_i(t) = \ap(t) +  \bm{x}_i(t)^\top \bm{\beta}(t) + \e_i(t), \quad t \in [0,1], \quad i=1,\dots,n,
		\end{align*}
		where $\bm{x}_i$ is a $p$-dimensional functional regressor with $\bm{x}_i(t) = (x_{i1}(t), \dots, x_{ip}(t))^\top$ for each $t \in [0,1]$. 
		With $\bm{\beta}(t) = (\beta_1(t), \dots, \beta_p(t))^\top$ for $t \in [0,1]$, we write $\bm{\beta}_1$ for the later part of the coefficient function $\bm{\beta}$ such that $\bm{\beta}_1(t) \equiv (\beta_{q+1}(t), \dots, \beta_p(t))^\top$ for each $t \in [0,1]$. 
		We may be interested in a partial test $H_0:\bm{\beta}_1 = \bm{0}$.
		For this testing, we may adopt the classical least squares estimator pointwise.
		To explain, we define the following quantities for each $t \in [0,1]$:
		\begin{align*}
			\bm{X}_0(t) \equiv [x_{ij}(t)]_{1 \leq i \leq n, 0 \leq j \leq q},
			\quad
			\bm{X}_1(t) \equiv [x_{ij}(t)]_{1 \leq i \leq n, q+1 \leq j \leq p},
			\quad
			\bm{X}(t) = \begin{bmatrix} \bm{X}_0(t) & \bm{X}_1(t) \end{bmatrix},
		\end{align*}
		where $x_{i0}(t) = 1$,
		along with projection matrix
		$\Pi_{\bm{X}_0(t)} \equiv \bm{X}_0(t)\{\bm{X}_0(t)^\top \bm{X}_0(t)\}^{-1} \bm{X}_0(t)^\top$
		and the orthogonal complement data matrix
		$\bm{X}_{1|0}(t) \equiv \bm{X}_1(t) - \Pi_{\bm{X}_0(t)} \bm{X}_1(t)$.
		Then, the (pointwise) least squares estimator $\hat{\bm{\beta}}_1$ for $\bm{\beta}_1$ is defined by
		\begin{align*}
			\hat{\bm{\beta}}_1(t) \equiv \{\bm{X}_{1|0}(t)^\top \bm{X}_{1|0}(t)\}^{-1} \bm{X}_{1|0}(t)^\top \bm{Y}(t), \quad t \in [0,1].
		\end{align*}
		The multivariate functional statistic for $H_0:\bm{\beta}_1 = \bm{0}$ is then $T_{\mathrm{fcr},n} \equiv \sqrt{n}\hat{\bm{\beta}}_1$. 
		This framework is studied by \cite{WPCL18, GM22}.
		In particular, it includes (i) testing in function-on-scalar regression models \citep{SF04, ZC07, Zhang11, CLO17} when we consider constant functions for all regressors $\{\bm{x}_i\}_{i=1}^n$, (ii) functional ANOVA setups for comparing several means as in \autoref{egFANOVA}(b) with suitable design matrix, and (iii) hence two-sample mean tests described in \autoref{egFANOVA}(a).

		\newcommand{\BB}{\mathsf{B}}
		
		\item \textit{Function-on-function regression (FoFR)}:
		\begin{align} \label{eq_fofr_model}
			Y_i = \ap + \BB X_i + \e_i, \quad i=1,\dots,n,
		\end{align}
		where $\BB:\HH\to\HH$ is a bounded linear operator on $\HH$ and $X_i$ is a random element that take values in $\HH$.
		Although direct inference for $\BB$ may not be feasible since CLT for $\BB$ itself may not exist \citep{CM13},
		we can conduct inference for mean response $\mu(X_0) \equiv \ap + \BB X_0$ at a new functional regressor $X_0$. 
		Namely, we may be interested in testing whether the mean response $\mu(X_0)$ at $X_0$ is equal to the global mean $\eo[Y]$,
		which can be represented by the null hypothesis $H_0:\mu(X_0) = \eo[Y]$ \cite[cf.][]{Yeon26}.
		This is equivalent to the statement $\theta \equiv \BB (X_0 - \eo[X]) = 0$.
		A functional statistic is then $T_{\mathrm{fofr},n} \equiv \sqrt{n/\hat{s}}\hat{B}(X_0-\bar{X})$ with suitable estimator $\hat{B}$ for $B$ and scaling $\hat{s}$.
		This problem is treated by 
		\cite{Yeon26} with $L^2$ norm and \cite{DT24} with sup norm. 
		See \autoref{ssec_fofr} for more details about the test statistic developed by \cite{Yeon26}.
	\end{enumerate}
\end{eg}

In what follows, we focus on the two-sample functional mean tests in \autoref{egFANOVA}(a) and the mean response inference for FoFR in \autoref{egFunReg}(b), 
whose bootstrap procedures are described in Sections~\ref{ssec_2test}--\ref{ssec_fofr}, respectively.

\subsection{Two-sample functional mean tests} \label{ssec_2test}

We explain a bootstrap method for two-sample inference for functional data.
Following the setup in \autoref{egFANOVA}(a), 
we consider the two-sample test statistic defined by
\begin{align} \label{eq_stat_two}
	T_{\mathrm{two},\bm{n}} = \sqrt{n_1n_2 \over n_1+n_2} (\bar{X}_1 - \bar{X}_2)
\end{align}
with $\bm{n} \equiv (n_1,n_2)$. 
For the resampling implementation of our depth-based two-sample test, we adopt the following residual-type bootstrap method; see \cite{PS16,BHK09,DKA20} for other bootstrap  two-sample functional mean tests.
%We point out that our method is generally applicable as long as the resampling method consistently approximates the sampling distribution of the statistic.
With the pooled sample mean $\bar{X} \equiv (n_1+n_2)^{-1} (n_1\bar{X}_1 + n_2\bar{X}_2)$, 
for each $k=1,2$,
the bootstrap resamples are defined by $X_{ki}^* \equiv \bar{X} + \e_{ki}^*$,
where bootstrap errors $\e_{k1}^*, \dots, \e_{kn_k}^*$ are iid uniform draws from the (centered) residuals $\{\hat{\e}_{ki} \equiv X_{ki}-\bar{X}_k \}_{i=1}^{n_k}$. 
The bootstrap statistic is then defined by
\begin{align} \label{eq_stat_two_bts}
	T_{\mathrm{two},\bm{n}}^* \equiv \sqrt{n_1n_2 \over n_1+n_2} (\bar{X}_1^* - \bar{X}_2^*),
\end{align}
where $\bar{X}_k^* \equiv n_k^{-1} \sum_{i=1}^{n_k} X_{ki}^*$ for $k=1,2$.
For completeness, we provide the asymptotic and bootstrap results for this functional two-sample framework along with its proof (\autoref{App2test} of the supplement), 
as we have not found an exact reference. 

\begin{thm} \label{thm_two}
	Suppose that $n_1\to\infty$ and $n_2\to\infty$ along with
	$n_1/(n_1+n_2) \to \zeta \in (0,1)$.
	We further assume that each random function $X_{ki}$ has finite second moment with covariance operator $\ga_k$. 
	By denoting $P_{\mathrm{two},\bm{n}}$ and $\widehat{P}_{\mathrm{two},\bm{n}}$ for the distributions of $T_{\mathrm{two},\bm{n}}$ in \eqref{eq_stat_two} and $T_{\mathrm{two},\bm{n}}^*$ in \eqref{eq_stat_two_bts}, respectively,
	under the null hypothesis $H_{0,\mathrm{two}}:\mu_1=\mu_2$, we have the following:
	\begin{enumerate}[(a)]
		\item $\pi(P_{\mathrm{two},\bm{n}}, P_{\mathrm{two},\zeta}) \to 0$, where $P_{\mathrm{two},\zeta}$ represents the Gaussian distribution on $\HH$ with mean zero and covariance $(1-\zeta)\ga_1 + \zeta\ga_2$; and
		
		\item $\pi(\widehat{P}_{\mathrm{two},\bm{n}}, P_{\mathrm{two},\bm{n}}) \to 0$ almost surely.
		
	\end{enumerate}
\end{thm}

\subsection{FoFR mean response inference} \label{ssec_fofr}

\newcommand{\BB}{\mathsf{B}}

We illustrate the depth statistic in the FoFR setting described in \autoref{egFunReg}(b).
For this, we follow the estimation and bootstrap inference procedure developed in \cite{Yeon26}. 
We observe $n$ functional response-regressor pairs $\{(Y_i, X_i)\}_{i=1}^n$ from the FoFR model in \eqref{eq_fofr_model},
and a new functional predictor $X_0 \ed X_1$ is given.
The sample covariance operator of $\{X_i\}_{i=1}^n$ and the sample cross-covariance operator between $\{Y_i\}_{i=1}^n$ and $\{X_i\}_{i=1}^n$ are respectively given by
\begin{align*}
	\widehat{\ga} \equiv n^{-1} \sum_{i=1}^n(X_i - \bar{X})^{\otimes 2}, 
	\quad
	\widehat{\Dt} \equiv n^{-1} \sum_{i=1}^n \{ (Y_i - \bar{Y}) \otimes (X_i - \bar{X}) \},
\end{align*}
where $\bar{X} \equiv n^{-1} \sum_{i=1}^n X_i$ and $\bar{Y} \equiv n^{-1} \sum_{i=1}^n Y_i$ represent the sample means. 
The covariance operator $\widehat{\ga}$ admits spectral decomposition \citep[Chapter~4]{HE15} as
$\hat{\ga} = \sum_{j=1}^n \hat{\g}_j (\hat{\phi}_j \otimes \hat{\phi}_j)$,
where $(\hat{\g}_j, \hat{\phi}_j)$ is the $j$-th eigenvalue-eigenfunction pair of $\widehat{\ga}$.
To estimate the slope estimator $\BB$,
we regularize the inversion of $\widehat{\ga}$ by truncating it up to the truncation level $J = J_n$.
Namely, with a truncated inverse covariance operator $\widehat{\ga}_J^{-1} \equiv \sum_{j=1}^J \hat{\g}_j^{-1} (\hat{\phi}_j \otimes \hat{\phi}_j)$, 
the functional principal component regression (FPCR) estimator $\widehat{\BB}_J$ for $\BB$ is defined by
\begin{align} \label{eq_fofr_est}
	\widehat{\BB}_J \equiv \widehat{\Dt} \widehat{\ga}_J^{-1}.
\end{align}
Upon scaling, the test statistic for $H_0:\mu(X_0) = \eo[Y]$ is then defined  by
\begin{align} \label{eq_fofr_stat}
	T_{\mathrm{fofr},n}(X_0) \equiv \sqrt{n \over \hat{\tau}_J(X_0)} \widehat{\BB}_J(X_0 - \bar{X}),
\end{align} 
where the scaling term is defined by
\begin{align*}
	\hat{\tau}_J(x) \equiv \sum_{j=1}^J \hat{\g}_j^{-1} \langle x- \bar{X}, \hat{\phi}_j \rangle^2, \quad x \in \HH.
\end{align*}

The following residual bootstrap is shown to be valid to  approximate the distribution of $T_n(X_0)$ in \eqref{eq_fofr_stat}.
We first compute the residuals as $\hat{\e}_{i,J_{\mathrm{res}}} \equiv Y_i - \bar{Y} - \widehat{\BB}_{J_{\mathrm{res}}}(X_i - \bar{X})$ for additional truncation level $J_{\mathrm{res}} = J_{\mathrm{res},n}$ with $\hat{\BB}_{J_{\mathrm{res}}} \equiv \widehat{\Dt} \widehat{\ga}_{J_{\mathrm{res}}}^{-1}$. 
Then, the bootstrap errors $\{\e_i^*\}_{i=1}^n$ are drawn uniformly from  $\{\e_{i,K}\}_{i=1}^n$ with replacement and we compute the bootstrap responses akin to \eqref{eq_fofr_model} as
\begin{align*}
	Y_i^* \equiv \bar{Y} + \widehat{B}_{J_{\mathrm{cen}}}(X_i - \bar{X}) + \e_i^*, \quad i=1,\dots,n,
\end{align*}
where $\widehat{\BB}_{J_{\mathrm{cen}}} \equiv \widehat{\Dt} \widehat{\ga}_{J_{\mathrm{cen}}}^{-1}$ plays a role of the true parameter in the bootstrap world with another extra truncation $J_{\mathrm{cen}} = J_{\mathrm{cen},n}$.
The bootstrap slope estimator $\widehat{\BB}_J^*$ is constructed analogously to \eqref{eq_fofr_est} as
\begin{align*}
	\widehat{\BB}_J^* \equiv \widehat{\Dt}^* \widehat{\ga}_J^{-1},
\end{align*}
where $\widehat{\Dt}^* \equiv n^{-1} \sum_{i=1}^n \{(Y_i^* - \bar{Y}^*) \otimes (X_i - \bar{X})\}$ with $\bar{Y}^* \equiv n^{-1} \sum_{i=1}^n Y_i^*$.
The bootstrap test statistic is then defined by
\begin{align} \label{eq_fofr_stat_bts}
	T_{\mathrm{fofr},n}^*(X_0) \equiv \sqrt{n \over \hat{\tau}_J(X_0)} (\widehat{\BB}_J^* - \widehat{\BB}_{J_{\mathrm{cen}}}) (X_0 - \bar{X}).
\end{align}

Regarding the mean response inference in this function-on-function regression framework, the following theorem is established by \cite{Yeon26}. 

\begin{thm} \label{thm_fofr}
	Write $P_{\mathrm{fofr},n}$ and $\widehat{P}_{\mathrm{fofr},n}$ for the distributions of $T_{\mathrm{fofr},n}(X_0)$ in \eqref{eq_fofr_stat} and $T_{\mathrm{fofr},n}^*(X_0)$ in \eqref{eq_fofr_stat_bts}, respectively.
	Under the assumptions stated in \autoref{AppFoFRcond} of the supplement, we have the following:
	\begin{enumerate}[(a)]
		\item $\pi(P_{\mathrm{fofr},n}, P_{\mathrm{fofr},\sa}) \to 0$,
		where $P_{\mathrm{fofr},\sa}$ denotes the Gaussian random function on $\HH$ with mean zero 0 and covariance operator $\sa \equiv \eo[\e \otimes \e|X]$; and
		
		\item $\pi(\widehat{P}_{\mathrm{fofr},n}, P_{\mathrm{fofr},n}) \xrightarrow{\pr} 0$.
		
	\end{enumerate}
	
\end{thm}

\section{Functional depth} \label{secFdepth}

%We refer to \cite{GN17} for a general review of their theoretical properties. 

We briefly review four representative functional depths: 
the kernel, regularized halfspace, integrated, and infimal depths. 
The first two are defined on a general separable Hilbert space $\HH$, 
and thus apply to general Hilbertian random elements \cite[cf.][]{JP20}. 
The latter two are defined on $C^d([0,1])$ and have also been extended to multivariate functional data, possibly on multivariate domains. 
We verify that Conditions~\ref{condContPoint}--\ref{condLimitCont} hold for the depth functions that we consider. 
Therefore, the depth statistic with any of these four functional depths is valid for the inference problems in Sections~\ref{ssec_2test}--\ref{ssec_fofr}.
%For the other depths, however, we do not establish the corresponding conditions required by our theory; 
%in addition, the last two are formulated on a Banach space rather than the Hilbert-space setting considered here. 
%Nevertheless, we include the associated depth statistics in our numerical studies as empirical benchmarks for broader comparison.
%Although it is unclear that Conditions~\ref{condContPoint}--\ref{condContProb} are satisfied for the second depth and the last two depths are defined on a Banach space,
%we will shortly review them and include the corresponding depth statistics in our numerical studies for a broader comparison.
%Other popular functional depths include 
%modified band \citep{LR09}
%spatial \citep{CC14b}
%extremal \citep{NN16}
%Mahalanobis \citep{BBC20}
%$L^2$ \citep{RC24}
%spherical \citep{MN25}
%depths, from classical to recent depth proposals.
%In this section, we will use $Q$ to represent generic notation for probability distributions on a function space.

\subsection{Kernel depth} \label{ssecFdepthKD}

We first introduce a functional depth which we call the kernel depth \citep{CCF06}.
Let $K:[0,\infty)\to[0,\infty)$ be a continuous and non-increasing function with $K(0)>0$ and $\lim_{t\to\infty} K(t) = 0$,
which we call a kernel (function).
Popular kernels include the Gaussian kernel given by $K(t) = e^{-t^2/2}$ for $t\in [0,\infty)$.
The \emph{kernel depth} (KD) of $x \in \HH$ (induced by $K$) with respect to a probability measure $Q$ on $\HH$ is defined by
\begin{align} \label{eqKDpop}
	D_K(x; Q) = \eo[K(\|x-X\|)]
\end{align}
where $X \sim Q$. 
\autoref{thmKD} verifies Conditions~\ref{condContPoint}--\ref{condLimitCont} under mild assumptions.

%\tred{
%The proof of Theorem 5 relies on the fact that the kernel depth can be represented
%as a kernel mean embedding. This is, however, not true for any “kernel” function K
%(as considered in Theorem 5). For the equivalence, it is necessary that K induces a
%reproducing kernel Hilbert space, which must be assumed. For this reason, the claim of
%Theorem 5 needs to be changed.}

\begin{thm} \label{thmKD}
	Under Conditions~\ref{condK1}--\ref{condK2} in \autoref{AppContKD}, the KD satisfies Conditions~\ref{condContPoint}--\ref{condContProb}. 
	If $P$ is a Gaussian distribution on $\HH$ with zero mean and positive-definite covariance operator, 
	when using the Gaussian kernel given by $K(t) = e^{-t^2/2}$ for $t \in [0,\infty)$,
	then it also satisfies Condition~\ref{condLimitCont}.
\end{thm}

The proof of \autoref{thmKD} is given in \autoref{AppContKD} of the supplement.
More specifically, \autoref{thmKDcont} provides uniform continuity results of the KD,
%basically indicates that the KD satisfies Conditions~\ref{condContPoint}--\ref{condContProb}. 
%These two results 
which have not been clearly established in the literature, 
although a uniform consistency by the empirical distribution is given in \citet[Theorem~5]{WN25}. 
It is important to note that the KD is useful especially for functional (infinite-dimensional) data,
because such convergence results of functional depths uniform on an entire space have been rarely found.
In addition, the continuity of the KD in probability measures is closely related to that in the so-called maximum mean discrepancy in Machine Learning literature \cite[cf.][]{SBSM23}. 
While we do not attempt to thoroughly study this direction, 
we refer to \cite{WN25} for more details about this relationship.

%We close this section by providing an important example where the KD additionally satisfies Condition~\ref{condLimitCont}.
%To summarize, using the KD, Theorems~\ref{thmDepthStat}--\ref{thmDepthPval} hold (i.e., the depth statistic is valid) for the inference problems in Sections~\ref{ssec_2test}--\ref{ssec_fofr}.

%\tred{
%Because the kernel depth can be used with any kernel function K, I do not see how, in
%Example 3, an explicit expression for the kernel depth can be obtained. Is Example 3
%meant only for the Gaussian kernel?}

Traditionally, $D_K$ has been called \emph{$h$-mode depth},
since its practical version involves a bandwidth $h \in (0,\infty)$.
Namely, upon defining a normalized kernel $K_h$ as $K_h(u) = h^{-1} K(h^{-1}u)$ for $u \in [0,\infty)$,
then the sample KD of $x \in \HH$ is defined by
\begin{align} \label{eqKDsample}
	D_{K_h}(x;\widehat{Q}_n)
	= (nh)^{-1} \sum_{i=1}^n K(h^{-1}\|x-X_i\|),
\end{align}
where $\widehat{Q}_n$ stands for the empirical distribution defined through $\{X_i\}_{i=1}^n$. 
Here, the bandwidth $h$ can be data-adaptively selected.
For example, for $u \in (0,1)$, we may take $u$-th quantile of the observed distances $\{\|X_{i'}-X_i\|: i \neq i'\}$,
where a KD with smaller $u$ tends to be sensitive to shapes of the functional observations  \citep{WN25}.
%In addition to its computational scalability, more advantages of the KD $D_K$ are recently found:
%it can characterize a distribution on $\HH$ and exhibit uniform consistency over the entire function space $\HH$ with a suitable kernel.
%While we direct readers to \cite{WN25} for a comprehensive review of the KD,
%we provide some useful results for our methodology. 

\subsection{Regularized halfspace depth} \label{ssecFdepthRHD}

It is well-known that a naive extension of the celebrated Tukey's halfspace depth \citep{Tukey75} can be degenerate in an infinite-dimensional space \citep{DGC11}.
To describe, we define the functional halfspace depth (FHD) $D_{\rm FHD}(x;Q)$ of $x \in \HH$ at $Q \in \pp(\HH)$ as
\begin{align*}
	D_{\rm FHD}(x;Q) = \inf_{v \in \HH:\|v\|=1} \pr(\langle X-x, v\rangle \geq 0),
\end{align*}
where $X \sim Q$. 
However, for example, when $Q$ is Gaussian, the FHD is turned out to be degenerate as $\pr(D_{\rm FHD}(T;Q) = 0) = 1$ for $T \sim Q$ independent of $X$ \cite[cf.][]{YDL25RHD}. 
In our context, unlike the KD (cf.~\autoref{egKDlimitCont}), 
the FHD may not satisfy Condition~\ref{condLimitCont}.

To address this degeneracy of the FHD, the regularized halfspace depth (RHD) was recently introduced by \cite{YDL25RHD}.
This generalizes the classical halfspace depth on the Euclidean space to a general Hilbert space by regularizing the dispersions of projections in the projection direction set.
Let $Q \in \pp(\HH)$, $X \sim Q$ with $\eo[\|X\|^2]<\infty$. 
We denote $\ga(Q) \equiv \eo[(X-\eo[X])^{\otimes 2}]$ for its covariance operator.
To reduce computational burden due to spectral decomposition used in the original RHD, 
we suggest using the following modified version:
for $Q \in \pp(\HH)$ and $x \in \HH$, 
\begin{align} \label{eqRHDmod}
	D_{\mathrm{RHD},\ld}(x;Q) 
	\equiv \inf_{v \in \HH:\|v\|=1, \|\ga(Q)^{1/2}v\| \geq \ld} \pr(\langle X-x, v \rangle \geq 0),
\end{align}
where $X \sim Q$. 
Here, $\ga(Q)^{1/2}$ stands for the square root operator of $\ga(Q)$;
see \cite{HE15} for general Hilbert space theory. 
Unlike the FHD exhibits degeneracy,
the RHD is shown to be non-degenerate by \cite{YDL25RHD},
which ensures its practicality.
For completeness, we provide a nondegeneracy result of the modified RHD \eqref{eqRHDmod} in \autoref{AppRHDnondege} of the supplement.
%In addition, the original RHD possesses many desirable depth properties including isometry invariance, vanishing at infinity, and uniform consistency over a totally bounded set.
See also \autoref{AppRHDimple} for its numerical implementation.

Under mild assumptions, this modified RHD satisfies Conditions~\ref{condContPoint}--\ref{condLimitCont} when the limiting distribution is Gaussian. 
The following theorem formalizes this result, with its proof deferred to \autoref{appRHDconsist}.
\begin{thm} \label{thmRHD}
	Suppose that $P$ is a Gaussian distribution on $\HH$ with positive-definite covariance operator $\ga(P)$ 
	and let $\ld \in (0,\|\ga(P)^{1/2}\|)$.
	Then, Conditions~\ref{condContPoint} and~\ref{condLimitCont} hold for the modified RHD in \eqref{eqRHDmod}. 
	If we further assume that $\opnorm{\ga(Q_n)-\ga(P)}_\infty \to 0$ holds for a sequence $\{Q_n\}$ in $\pp(\HH)$ weakly converging to $P$, 
	then Condition~\ref{condContProb} also holds. 
\end{thm}

\subsection{Integrated and infimal depths} \label{ssecFdepthID}

The last two functional depths for continuous functional data we introduce are integrated \citep{FM01} and infimal \citep{Mosler13} depths.
These pointwise apply depth functions for univariate data and integrate these by taking either integral or infimum.
To explain, let $D_{\mathrm{multi}}$ be a multivariate depth function on $\R^d$, such as halfspace \citep{Tukey75} or simplicial \citep{Liu90} depths, used for constructing functional depths. 
The integrated (ITD) and infimal (IFD) depths of $x \in C^d([0,1])$ with respect to $Q \in \pp( C^d([0,1]) )$ are defined respectively by
\begin{align*}
	D_{\mathrm{ITD}}(x;Q) & \equiv \int_0^1 D_{\mathrm{multi}} (x(u);Q_u) du, \\
	D_{\mathrm{IFD}}(x;Q) & \equiv \inf_{u \in [0,1]} D_{\mathrm{multi}} (x(u);Q_u).
\end{align*}
Their sample versions are defined by replacing the probability $Q$ by its empirical counterpart $\widehat{Q}_n$.
%infimal tie breaks extremal depth \cite{NN16}
%Uniform consistency results are established by \cite{NGOH16} and \cite{GN15} respectively for the ITD and IFD. 
The implementation of both depth functions are well-developed in R package \texttt{ddalpha} \citep{ddalpha}.
While many other functional depths have been proposed, these classical depths play benchmark roles in the literature.
These depths can be extended to multivariate functional data defined possibly on multivariate domain \cite[e.g.,][]{CHSV14, QDG22}.

Under mild assumptions, 
we verify Conditions~\ref{condContPoint}--\ref{condLimitCont} for ITD and IFD, 
when $D_{\mathrm{multi}}$ is chosen as the halfspace depth $D_{\mathrm{HD}}$ on $\R^d$ defined by 
\begin{align*}
	D_{\mathrm{HD}}(\bm{x};\nu) 
	\equiv \inf_{\bm{v} \in \R^d:\|\bm{v}\|_{\R^d}=1} \pr(\bm{X}^\top \bm{v} \geq \bm{x}^\top \bm{v}),
	\quad \bm{x} \in \R^d, \quad \bm{X} \sim \nu. 
\end{align*}
%To sum up, if $P$ is Gaussian and the covariance process is pointwise positive-definite, 
%the ITD and IFD (the latter only for $d=1$) satisfy all these conditions under mild assumptions. 
All technical conditions and details for Theorems~\ref{thmITD}--\ref{thmIFD} are deferred to \autoref{appID}.

\begin{thm} \label{thmITD}
	Suppose that $P_u$ has measure zero on all hyperplanes in $\R^d$ for each $u \in [0,1]$. 
	Then, the ITD satisfies Conditions~\ref{condContPoint}--\ref{condContProb}.
	Also, under Conditions~\ref{condGau} and~\ref{condITD1}--\ref{condITD3},  Condition~\ref{condLimitCont} holds for the ITD.
	
\end{thm}

\begin{thm} \label{thmIFD}
	Under Conditions~\ref{condGau} and~\ref{condIFD1}, 
	the IFD satisfies Conditions~\ref{condContPoint} and~\ref{condLimitCont}. 
	Also, when $d=1$, if $P_u$ is absolute continuous, then Condition~\ref{condContProb} holds for the IFD. 
%	\hfill
%	\begin{enumerate}[(a)]
%		\item Under Condition~\ref{condGau}, if $\sa$ is continuous, then $D_{\mathrm{IFD}}(\cdot;P)$ is Lipschitz continuous and hence satisfies Condition~\ref{condContPoint}.
%		
%		\item 
%		When $d=1$, if $P_u$ is absolute continuous, then Condition~\ref{condContProb} holds for the IFD. 
%		
%		\item Under Condition~\ref{condGau}, the IFD satisfies Condition~\ref{condLimitCont}. 
%	\end{enumerate}
\end{thm}

\begin{rem}
	Theorems~\ref{thm_two}--\ref{thm_fofr} 
	verify Conditions~\ref{condConvDistnData}--\ref{condConvDistnBTS} in a general Hilbert-space framework. 
	In particular, they apply to 
	the Sobolev space $\W^{1,2}([0,1]) = \{f \in L^2([0,1]) : f' \in L^2([0,1])\}$, 
	%	of once differentiable functions on $[0,1]$ whose derivatives are also square-integrable,
	equipped with inner product $\langle f,g \rangle = \int_0^1 f(u)g(u),du + \int_0^1 f'(u)g'(u),du$. 
	The Sobolev embedding theorem \cite[cf.][Lemma~D.3]{Yeon26GBA}, together with the continuous mapping theorem, then implies that these conditions for the two-sample mean and FoFR mean response problems also hold on \((C([0,1]), \|\cdot\|_{\sup})\).
	Therefore, under the assumptions of Theorems~\ref{thmITD}--\ref{thmIFD}, 
	the ITD and IFD on $C([0,1])$ are also applicable to these inference problems.
\end{rem}

\section{Numerical studies} \label{sec5}

All functional observations are generated based on the Karhunen--Lo\'{e}ve expansion:
\begin{align} \label{eqXgenKL}
	X = \mu + \sum_{j=1}^{J_{\mathrm{true}}} \g_j^{-1/2} \xi_j \phi_j
\end{align}
for $\mu \in \HH$.
Here, the number $J_{\mathrm{true}}=20$ of true eigenfunctions are fixed for all simulation studies.
The eigenvalues $\{\g_j\}_{j=1}^{J_{\mathrm{true}}}$, eigenfunctions $\{\phi_j\}_{j=1}^{J_{\mathrm{true}}}$, functional principal component (FPC) scores $\{\xi_j\}_{j=1}^{J_{\mathrm{true}}}$ will vary depending on cases.
The FPC scores are distributed as $\xi_j \ed \xi W_j$,
where $W_j \iid \nd(0,1)$ while the distribution of the latent variable $\xi$ is given by either $\xi =1$, $\xi \sim \nd(0,1)$, or $\xi \sim \mathsf{Exp}(1)-1$. 
We refer to these as N1-, NN-, NE-types of FPC scores, respectively.
The former generates a Gaussian element with independent FPC scores, 
while in the latter two cases, the scores $\{\xi_j\}$ are dependent, which could lead to a difficult inference case \cite[cf.][]{YDN23RB}.
The eigenvalues are defined through their gaps $\g_j-\g_{j+1} = 2j^{-a}$ with $\g_1 \equiv 2\sum_{j=1}^\infty j^{-a}$ and determined by the decay rate $a > 2$.
For eigenfunctions, we consider subsets $\{\phi_{q,j}\}_{j=1}^{J_{\mathrm{true}}}$ of four orthonormal systems for $q \in \{\rm tri, mono, cheb, spl\}$ in $\HH = L^2([0,1])$: trigonometric functions, and orthonormalized monomials, (shifted) Chebyshev polynomials, and cubic splines functions.
The details of the constructions of these basis functions are deferred to \autoref{AppOrtho} of the supplement.
All functions are evaluated 50 equi-distant grid in [0,1].

For depth statistics implementations based on IFD and ITD, we use well-implemented functions in the \texttt{ddalpha} package in R \citep{ddalpha}. 
The regularization $\ld$ for the RHD \eqref{eqRHDmod} and bandwidth for the KD in \eqref{eqKDsample} are determined 
by quantile levels $u \in \{0.5, 0.4, 0.3, 0.2, 0.1, 0.01, 0.001\}$ 
as suggested by \cite{YDL25RHD} and \cite{WN25}, respectively.
We found that, in both sizes and power,
$u=0.1$ and $u=0.001$ work best for the RHD statistic in two-sample and FoFR problems, respectively,
while the KD statistic performs well with $u=0.01$ in both problems.
We only report the results from these quantile levels.
In addition, to approximate the infimum in the RHD, we employ a random projection approach as in the original paper.
Since the RHD may produce ties, we devise a new tie breaking algorithm and apply it to our numerical studies. 
The RHD implementation is further detailed in \autoref{AppRHDimple} of the supplement.

All tests are conducted at significance level 0.05.
We conducted 1,000 Monte Calro iterations to compute empirical rejection rates, which provide empirical sizes under a null hypothesis and empirical power under an alternative hypothesis. 
Bootstrap methods are implemented with 1,000 bootstrap resamples.
Since the results are similar across scenarios for each problem, we only partially report them.

\subsection{Two-sample functional mean tests} \label{ssec_5_1}

\subsubsection{Simulation design}

Two independent groups $\{X_{ki}\}_{i=1}^{n_k}$, $k=1,2$, of iid functional data are generated based on \eqref{eqXgenKL}:
\begin{align*}
	X_{ki} \equiv \mu_k + \sum_{j=1}^{J_{\mathrm{true}}} \g_{k,j}^{1/2} \xi_{ki,j} \phi_{k,j}, \quad i=1,\dots,n_k, \quad k=1,2.
\end{align*}
The total sample sizes are given by $n \equiv n_1+n_2 \in \{50, 200, 1000\}$, where the sample sizes of each group are equal as $n_1=n_2=n/2$.
The FPC scores are distributed as $\xi_{ki,j} \ed \xi_j$, where $\xi_j$ is from \eqref{eqXgenKL}.
We vary eigenvalues and eigenfunctions to consider different scenarios with equal and unequal covariance operators. 
Namely, when eigenvalues are equal, $\g_{k,j} - \g_{k,j+1} \equiv 2j^{-a_k}$ with $a_1=a_2 = 2.5$;
if eigenvalues differ across groups, $a_1=5$ and $a_2=2.5$.
For scenarios with equal eigenfunctions, $\phi_{k,j} = \phi_{\mathrm{tri},j}$ for both $k=1,2$,
while when eigenfunctions are unequal, we set $\phi_{1,j} = \phi_{\mathrm{mono},j}$ and $\phi_{2,j} = \phi_{\mathrm{cheb},j}$. 
Similar constructions can be found in \cite{YK26},
where the details are given in \autoref{AppOrtho} of the supplement.

While the mean function $\mu_1$ of the first group is set to be constant as $\mu_1 = 0 \in \HH = L^2([0,1])$,
we consider seven functions with different shapes for the mean function of the second group under the strongest alternatives,
inspired by outlier constructions considered in \citet[Section~4.1]{YDL25RHD}. 
These respectively represent magnitude, jump, peak, linear, quadratic, cubic, and wiggle differences between mean functions.
These are denoted by $\mu_{2,d}$ with index $d \in D_{\mathrm{alter}} \equiv \{\rm Mag, Jump, Peak, Lin, Quad, Cub, Wig\}$.
We depict these seven curves in \autoref{fig_2test_alter},
while the detailed formulae are given in \autoref{AppAlter} of the supplement.
It is worthwhile mentioning that such features cannot be considered in finite-dimensional cases. 
The true mean $\mu_2$ of the second group is then defined by $\mu_2 = (1-c)\mu_1+ c\mu_2 = c\mu_{2,d}$ with $c \in \{0,0.2, 0.4, 0.6, 0.8, 1\}$.
Here, $c=0$ renders the null hypothesis $H_0:\mu_1 = \mu_2$,
while the mean difference becomes larger as $c$ increases.

\begin{figure}[h!]
	\centering
	\includegraphics[width=0.7\linewidth]{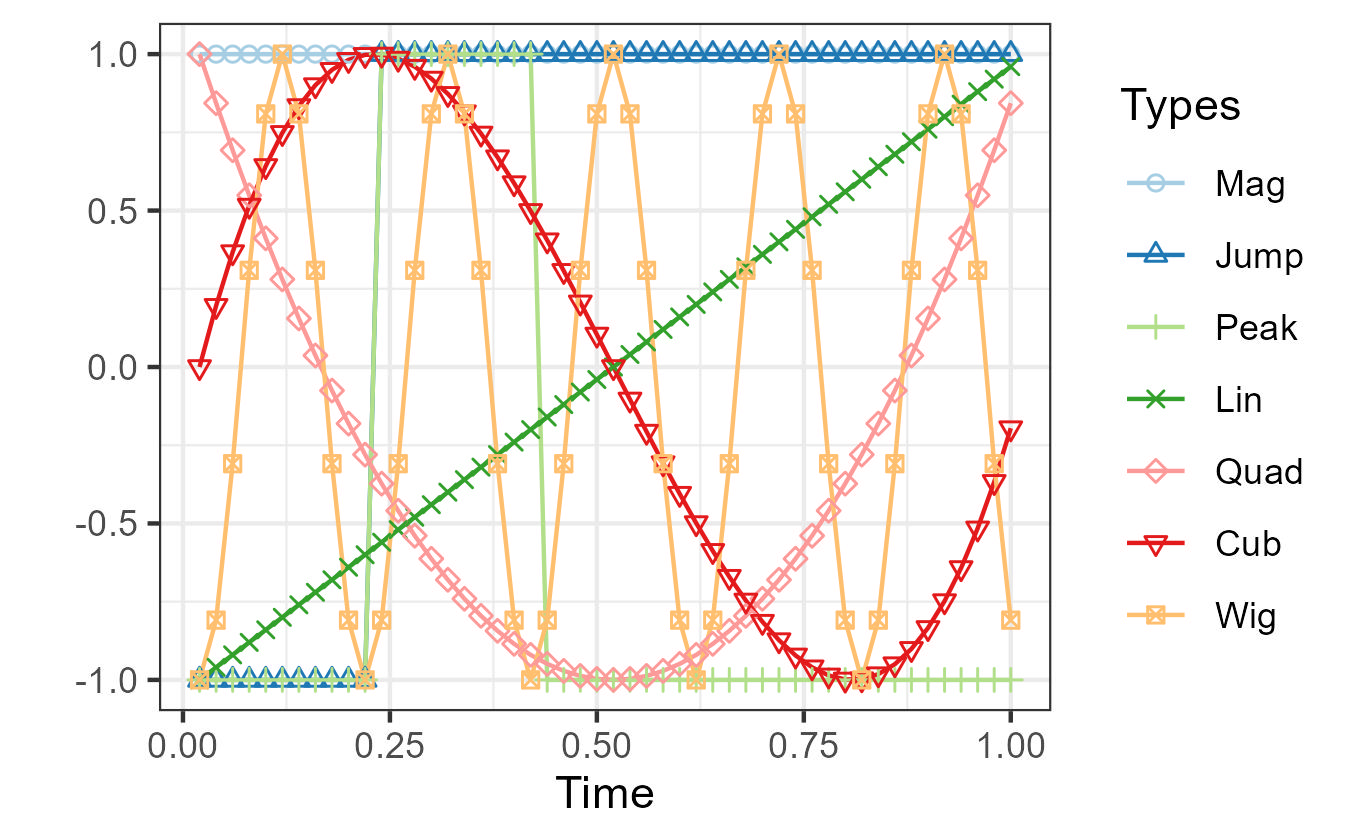}
	\caption{
		The curves with different shapes considered under strongest alternatives. 
	}
	\label{fig_2test_alter}
\end{figure}

\subsubsection{Competing tests}

In addition to the depth statistics based on ITD, IFD, RHD, and KD,
we consider several different state-of-the-art competing methods for two-sample functional mean tests including functional ANOVA approaches.
The first group of competing methods are: 
the quadratic forms by \cite{HKR13,SH18} (QuadInv and QuadMul, respectively); 
the $L^2$ statistic  by \cite{BHK09, HKR13, PS16} (L2); 
the supremum statistic by \cite{DKA20} (sup);
the global (integrated) pointwise F statistic by \cite{ZL14} (Fint);
and the supremum pointwise F statistic by \cite{ZCWZ19} (Fmax).
\autoref{AppFVE2test} of the supplement includes more details about the choice of truncation level for QuadInv and QuadMul.
We calibrate the critical values for these test statistics using the same bootstrap method, the one described in \autoref{ssec_2test}.
These are deserved to be called our main competitors as they are calibrated using the same bootstrap method and it will better show the effectiveness of the proposed depth statistic. 
Nevertheless, we also conduct the bootstrap for Fmax given in \cite{ZCWZ19} (Fmaxb) for comparison,
which is implemented by \texttt{fdANOVA::fanova.tests} in R with option \texttt{test="Fmaxb"}. 
It is important to note that the tests including Fmaxb and those described below may differ not only in their choice of test statistic but also in the resampling schemes used to approximate the sampling distributions.
Consequently, although we include some comparison with these other tests, our main goal is a comparison between test statistics, not overall tests.

For the methods based on FPCs, where truncation may be involved, we consider truncation levels in $\{1, \dots, 20\}$.
The truncation levels are chosen using fraction of variance explained (FVE), which is one of the most common approaches to select truncation level, with threshold 0.85 \cite[e.g,][Section~12.2]{KR17}.

Beyond the aforementioned tests based on the bootstrap method in \autoref{ssec_2test}, other competitors include permutation tests 
based on random projections by \cite{CF10} (TRP), 
basis function representation by \cite{GS15} (FP),
a functional version of the Hotelling's $T^2$ statistic by \cite{PSV18} (HotSq),
and graphical approach by \cite{MMJH20} (GET).
These are implemented in \texttt{fdANOVA} \citep{GS29fdANOVA} (with options \texttt{test="TRP"} and \texttt{test="FP"}), \texttt{fdahotelling},\footnote{https://github.com/astamm/fdahotelling} and \texttt{GET} \citep{MM24} packages, respectively.

In addition, we compare our depth statistics with confidence bands while treating the differences $\{Z_i \equiv X_{1i}-X_{2i}\}_{i=1}^{n/2}$ as one-sample. 
the naive pointwise t-test without adjusting for multiple testing (naive.t),
and bands for mean functions proposed by \cite{Degras11} (Bs), \cite{CR18} (BEc),
and \cite{LR23} (FFSCB.t),
which are implemented in \texttt{ffscb}\footnote{https://github.com/lidom/ffscb} package.
In particular, regarding the fairness parameter for FFSCB.t, we use 1 as implemented in the package.

Finally, as a most recent method, we include high-dimensional bootstrap approach that is applicable to functional ANOVA problem, developed by \cite{LLM23} and implemented in the package
\texttt{hdanova}.\footnote{https://github.com/linulysses/hdanova}
Instead of the trigonometric functions $\{\phi_{\mathrm{tri},j}\}$ onto which functional observations are projected in their method, 
we modify their function by considering the orthonormalized cubic spline functions $\{\phi_{\mathrm{spl},j}\}$ for the basis choice. 
The tuning parameter for their partial standardization, we consider both zero and the default implementation.
These are referred to as HANOVA0 and HDANOVA, respectively.

\subsubsection{Empirical sizes and power}

The empirical sizes of these 22 tests are given in \autoref{tb_2test_sizes}, when $n=50$ and the NN-type FPC scores are considered.
Overall, including the depth statistics, many tests are able to keep the nominal level regardless of covariance differences. 
QuadInv, QuadMul, Fint, Fmax, TRP, Bs, naive.t provide conservative results.
HotSq and FFSCB.t tend to fail to achieve the nominal level;
as expected, HotSq is only applicable for equal covariance scenarios.
FFSCB.t involves the fairness parameter, 
where the resulting power along with sizes becomes smaller as this parameter increases,
which is observed by \cite{LR23}.

\begin{table}[h!]
	\renewcommand{\arraystretch}{1.3}
	\centering
	\caption{
		Empirical sizes of two-sample functional mean tests,
		when $n=50$ and the FPC scorse are of NN-type.
		Sizes within [0.3, 0.7] are indicated in \textit{italic}.
	}
	\label{tb_2test_sizes}
	\medskip
	\begin{tabular}{c|cc|cc} \hline
		Eigenfunctions& \multicolumn{2}{c|}{Equal} & \multicolumn{2}{c}{Unequal} \\ \hline
		Eigenvalues & Equal & Unequal & Equal & Unequal \\ \hline
		ITD      & \textit{0.044}  & \textit{0.052}    & \textit{0.038}  & \textit{0.056}    \\
		IFD      & \textit{0.056}  & \textit{0.055}    & \textit{0.064}  & 0.082    \\
		RHD      & \textit{0.060}  & \textit{0.054}    & 0.074  & \textit{0.036}    \\
		%		RPD      & \textit{0.055}  & \textit{0.049}    & \textit{0.048}  & \textit{0.039}    \\
		KD       & \textit{0.069}  & \textit{0.067}    & \textit{0.064}  & \textit{0.056}    \\ \hline
		QuadInv  & 0.023  & 0.017    & 0.025  & 0.020    \\
		QuadMul  & 0.023  & 0.026    & 0.013  & \textit{0.033}    \\
		L2 & \textit{0.043}  & \textit{0.049}    & \textit{0.033}  & \textit{0.056}    \\
		sup  & \textit{0.035}  & \textit{0.038}    & \textit{0.036}  & \textit{0.056}    \\
		Fint     & 0.024  & 0.022    & 0.017  & 0.024    \\
		Fmax     & 0.015  & 0.016    & 0.011  & 0.011    \\ \hline
		Fmaxb    & \textit{0.037}  & \textit{0.045}    & \textit{0.049}  & \textit{0.065}    \\ \hline
		TRP      & 0.017  & 0.024    & 0.025  & 0.029    \\
		FP       & \textit{0.049}  & \textit{0.052}    & \textit{0.039}  & \textit{0.062}    \\
		HotSq    & \textit{0.050}  & 0.651    & 0.127  & 0.715    \\
		GET      & \textit{0.043}  & \textit{0.056}    & \textit{0.052}  & \textit{0.056}    \\ \hline
		FFSCB.t  & 0.078  & 0.086    & 0.091  & 0.097    \\
		Bs       & 0.002  & 0.004    & 0.002  & 0.010    \\
		BEc      & 0.414  & 0.295    & 0.414  & 0.404    \\
		naive.t  & 0.021  & 0.020    & 0.028  & 0.029    \\ \hline
		HDANOVA0     & \textit{0.042}  & \textit{0.039}    & \textit{0.030}  & 0.028    \\
		HDANOVA & \textit{0.053}  & \textit{0.043}    & \textit{0.047}  & \textit{0.051}    \\ \hline
	\end{tabular}
\end{table}

We now turn to our attention to power analysis, which is more important in our simulation studies.
The empirical rejection rates are provided in \autoref{fig_2test_power} under the same scenario.
Since the power curves overlap, which prevents visibility, we report the proposed depth statistics (red), our main competitors (QuadInv, L2, sup, Fint, Fmax, marked in blue) and some selective competitors that show higher powers (TRP, GET, HDANOVA, marked in green);
we exclude QuadMul since it shows almost zero power.
%The latter two groups of statistics are referred to as blue and green statistics, respectively.

The numerical performance of our proposed depth statistic is remarkable, especially when it is based on the KD \eqref{eqKDsample}. 
It shows highest or comparable power over all scenarios.
In particular, when shape alternatives are considered, the KD statistic is substantially powerful compared to the others.
Among the depth statistics, following the KD, 
the RHD and IFD statistics stand in second place and generally perform comparably with the green statistics.
The RHD mostly outperforms the IFD, except for the wiggle alternative scenarios.
This could be anticipated as the RHD shows less sensitivity to functions of high frequency like $\mu_{2,\mathrm{Wig}}$.
The ITD essentially works well only when differences appear in magnitude (i.e., under magnitude, jump, and peak alternatives).
This can be interpreted analogously to what happened in functional outlier detection;
since the ITD is known to be weak in detecting shape outliers,
the corresponding depth statistic could exhibit worse performance in rejecting shape alternatives.
We refer to \cite{YDL25RHD} for discussion of how a functional depth benefits functional outlier detection. 

\begin{sidewaysfigure}[p]
	\centering
	\includegraphics[width=0.99\linewidth]{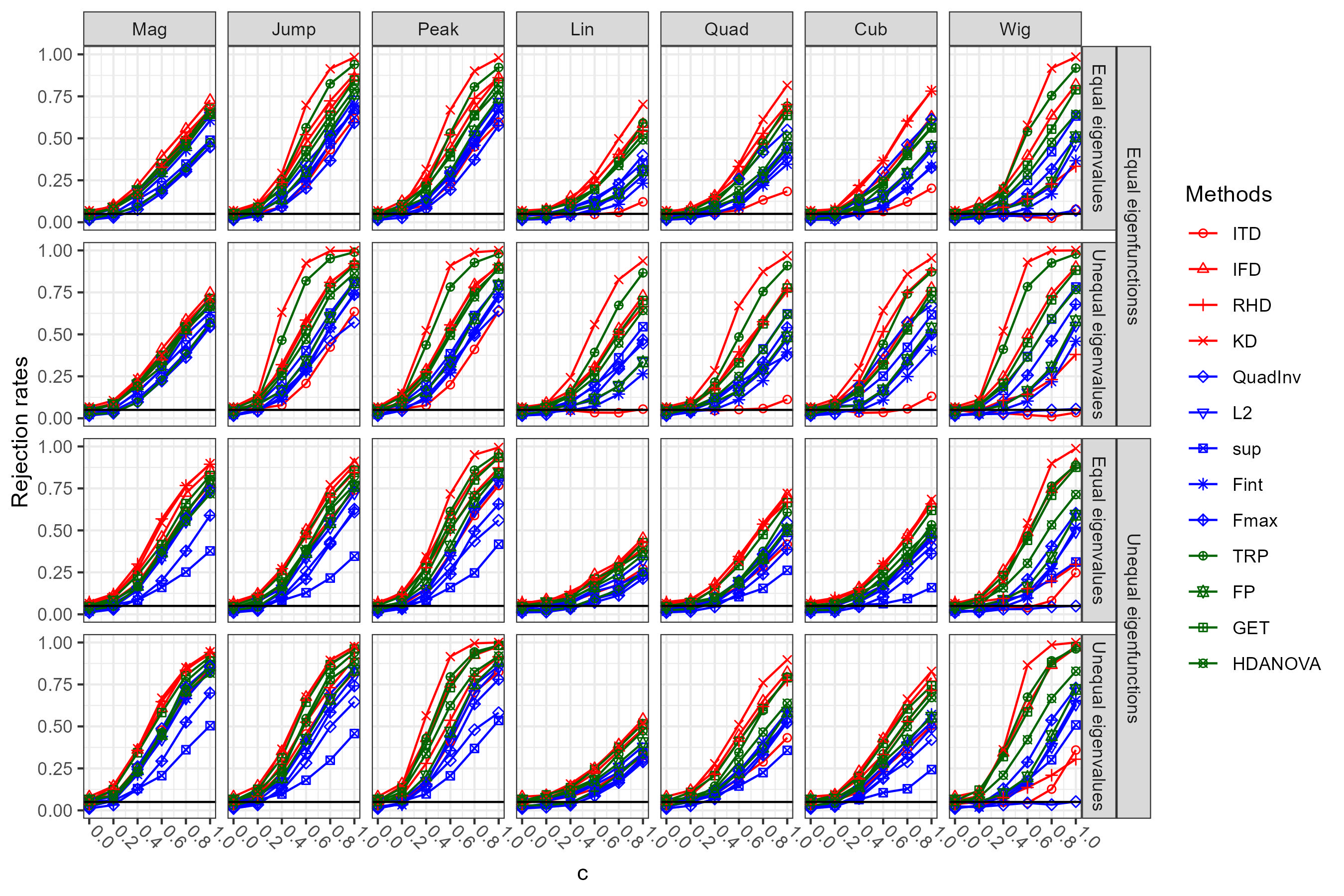}
	\caption{
		Empirical rejection rates of two-sample functional mean tests,
		when $n=50$ and the FPC scorse are of NN-type.
	}
	\label{fig_2test_power}
\end{sidewaysfigure}

The main competitors tend to behave similarly;
all of them perform worse than the KD, RHD, and IFD statistics 
under most scenarios.
The power of these blue statistics heavily relies on the equality of covariance operators and alternative scenarios.
For instance, in general, Fint performs worst when eigenfunctions are equal,
while the power of sup statistic is the lowest in the other cases with unequal eigenfunctions.
QuadInv works reasonably under magnitude, jump, and peak alternatives (even when covariance operators are different),
but its performance becomes worse as the form of $\mu_2$ is complicated;
for $\mu_{2,\mathrm{Wig}}$, it shows almost zero power.
The other competing tests are reasonably powerful in general,
although they perform worse than the KD statistic.
In most cases, these green statistics work better than the blue statistics

The overall winner of this power analysis for two-sample problem is the proposed depth statistic based on the KD \eqref{eqKDsample} with quantile level $u=0.01$.

\subsection{FoFR mean response inference} \label{ssec_5_2}

We continue our numerical analysis but under the FoFR setting in \autoref{ssec_fofr}.
Similarly to \eqref{eqXgenKL}, the regressors $\{X_i\}_{i=1}^n$ and $\{\e_i\}_{i=1}^n$ are generated based on the KL expansion as
$X_i = \sum_{j=1}^{J_{\mathrm{true}}} \g_j^{1/2} \xi_{ij} \phi_j$
and
$\e_i = \sum_{j=1}^{J_{\mathrm{true}}} \dt_j^{1/2} \zeta_{ij} \psi_j$,
respectively.
The eigenvalues are determined through their gaps $\g_j-\g_{j+1} = 2j^{-a_X}$ and $\dt_j-\dt_{j+1} = 2j^{-a_\e}$ with $\g_1 = 2\sum_{j=1}^\infty j^{-a_X}$ and $\dt_1 = 2\sum_{j=1}^\infty j^{-a_\e}$ for some $a_X \in \{2.5, 3.5\}$ and $a_\e \in \{2.5, 3.5\}$,
while the eigenfunctions are given as $\phi_j = \phi_{\mathrm{mono},j}$ and $\psi_j = \phi_{\mathrm{cheb},j}$.
We set $\xi_{ij} \ed \xi_j$ and $\zeta_{ij} \ed \xi_j$, where only NN- and NE-type FPC scores are considered.
The responses $\{Y_i\}_{i=1}^n$ are then computed based on the FoFR model \eqref{eq_fofr_model} with zero intercept $\ap = 0$ and slope operator $\BB$ described below.
The sample sizes under consideration are $n \in \{50, 200, 1000\}$

To construct the slope operator $\BB$, 
we first consider independent binary random variables $\{W_{\mathrm{slope},j}\}$ such that $\pr(W_{\mathrm{slope},j}=1) = 1/2 = \pr(W_{\mathrm{slope},j}=-1)$.
With decay rate $b \in \{1.5, 2.5\}$ and $J_0 = 5$, 
the true slope operator $\BB = \BB_c$ is given by 
\begin{align*}
	\BB 
	\equiv (1-c) \sum_{j=J_0+1}^{J_0} (2j^{-b}) W_{\mathrm{slope},j} \phi_{\mathrm{mono},j}^{\otimes2}
	+ c \sum_{j=1}^{2J_0} (2j^{-b}) W_{\mathrm{slope},j} \phi_{\mathrm{tri},j}^{\otimes2}
\end{align*}
for $c \in \{0, 0.2, 0.4, 0.6, 0.8, 1\}$.
The new regressor is generated also based on a KL expansion as $X_0 \equiv \sum_{j=1}^{J_{\mathrm{true}}} \g_{0j}^{1/2} \xi_{0j} \phi_{\mathrm{mono},j}$ independently of $\{(Y_i,X_i)\}_{i=1}^n$,
where $\g_{0j} = \g_j\I(j \leq J_0)$ and $\xi_{0j} \ed \xi_j$.
We note that, when $c=0$, it holds that $BX_0 = 0 = B\eo[X_1]$, and hence it renders the null hypothesis $H_0: \mu(X_0) = \eo[Y]$. 

For the residual bootstrap in \eqref{ssec_fofr}, three truncation levels $J, J_{\mathrm{res}},J_{\mathrm{cen}}$ are involved.
Following the previous studies \citep{YDN23RB, Yeon26}, we choose $J_{\mathrm{res}}$ as the one minimizing prediction error estimated by the 5-fold cross-validation, and set $J_{\mathrm{cen}}=J_{\mathrm{res}}$.
We then select $J$ as the minimum of the truncation levels bigger than $J_{\mathrm{cen}}$ whose FVE exceed the threshold 0.85.
This procedure is described in \autoref{AppDetailFoFR} of the supplement.

For depth statistics, we only consider the RHD and the KD,
as the frameworks for them and the FoFR residual bootstrap are commonly based on a Hilbert space. 
Since there are not many competing methods in this problem,
we compare the RHD and KD statistics with L2 and supremum norms used by \cite{Yeon26} and \cite{DT24}, respectively.
Since all tests achieve the nominal level well, 
we only provide their power curves.

%\autoref{tb_fofr_sizes} shows the empirical sizes of the tests,
%when $a_X = a_\e =2.5$, $b=1.5$, and the FPC scorse are of NN-type.
%Here, all tests achieves

%\begin{table}[h!]
%	\renewcommand{\arraystretch}{1.3}
%	\centering
%	\caption{
%		Empirical sizes of FoFR mean response inference,
%		when $a_X = a_\e =2.5$, $b=1.5$, and the FPC scorse are of NN-type. 
%	}
%	\label{tb_fofr_sizes}
%	\medskip
%	\begin{tabular}{c|cccc} \hline
%		& RHD   & KD    & L2    & sup   \\ \hline
%		$n=50$   & 0.042 & 0.054 & 0.044 & 0.049 \\
%		$n=200$  & 0.052 & 0.050 & 0.049 & 0.049 \\
%		$n=1000$ & 0.044 & 0.051 & 0.055 & 0.051 \\ \hline
%	\end{tabular}
%\end{table}

\autoref{fig_fofr_power} shows empirical power of the tests under consideration,
when $n=50$, $a_X = a_\e =2.5$, $b=1.5$, and the FPC scorse are of NN-type.
As in the two-sample problem, 
the depth statistics perform best.
Notably, the RHD here performs slightly better than the KD.
Another observations is that these depth statistics along with L2 statistic significantly outperforms the supremum statistic, 
where the latter can provide a simultaneous confidence band for $\BB$. 
This may indicate that, like confidence bands for two-sample problem are not powerful, 
the ones for FoFR may not be effective for hypothesis testing.

\begin{figure}[h!]
	\centering
	\includegraphics[width=0.7\linewidth]{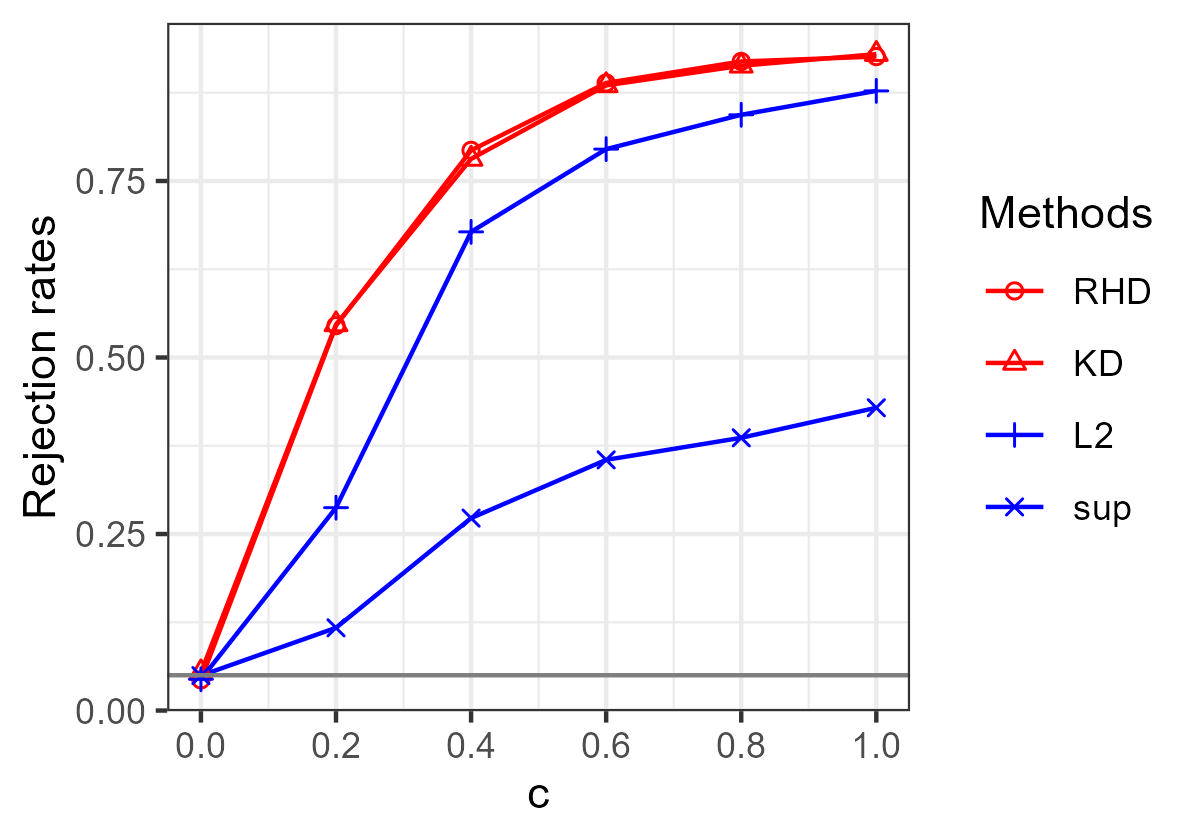}
	\caption{
		Empirical rejection rates for testing $H_0:\mu(X_0) = \eo[Y]$ in FoFR models,
		when $n=50$, $a_X = a_\e =2.5$, $b=1.5$, and the FPC scorse are of NN-type.
	}
	\label{fig_fofr_power}
\end{figure}

\subsection{Summary and guideline}

In this section, upon summarizing our discoveries, we provide a concise guideline for using depth statistics based on our comprehensive numerical studies.
\begin{enumerate}[(i)]
	\item 
	The KD statistics with $u=0.01$ works best among all tests and in all scenarios for both two-sample problem and FoFR. 
	\item 
	The RHD statistic also works well in general and particularly for FoFR, but it may be weak in detecting extremely wiggle alternatives. 
	The IFD statistic work well in general as well.
	Both RHD and IFD are comparable with other resampling-based methods (TRP, FP, HotSq, GET, HDANOVA0, HDANOVA).
	\item 
	The classical summary statistics (QuadInv, QuadMul, L2, sup, Fint, Fmax) are less effective to detect complex forms of alternatives.
	\item 
	Using confidence bands (e.g., FFSCB.t, BS, BEc, naive.t) may not be effective for hypothesis testing.
\end{enumerate}
Based on these summarized observations,
we recommend the depth statistic based on the KD with quantile level $u=0.01$.
%\tred{
%Given the limited scope of the simulation study and a large number of tuning parameters
%(the choice of the depth, the kernel, the bandwidth etc), the overall conclusions of the
%simulation study (we recommend the depth statistic based on the KD with quantile level
%u = 0.01) do not appear convincing.}
%This effectiveness of the KD in our proposed depth statistic may be interpreted by using the recent theoretical findings by \cite{WN25}: 
%the KD fully characterizes the distribution on a separable Hilbert space;
%and the sample KD uniformly converges to its population counterpart over the entire Hilbert space. 

\section{Data examples} \label{sec6}

\subsection{Two-sample tests with Canadian weather data}

We analyze the Canadian weather datasets.
These are available from the Environment  and Climate Change Canada website,
and extracted using \texttt{weathercan} R package \citep{weathercan}. 
We collect daily mean temperatures for years from 1962 to 2011, from  stations 889, 2205, 5097, 5415, 6358 located in cities, Vancouver, Calgary, Toronto, Montreal, Halifax, respectively.
Motivated by the abrupt increase in global average temperature around 1981 \cite[cf.][]{NOAA},
for each city, we form two groups as the temperature curves of the first twenty years (1962-1981) and the later thirty years (1982--2011), yielding $\{X_{1i}\}_{i=1}^{n_1}$ and $\{X_{2i}\}_{i=1}^{n_2}$ with $n_1=20$ and $n_2=30$, respectively. 
We are interested in whether the mean temperature curves differ between these two groups. 

\begin{figure}[!h]
	\centering
	\includegraphics[width=0.99\linewidth]{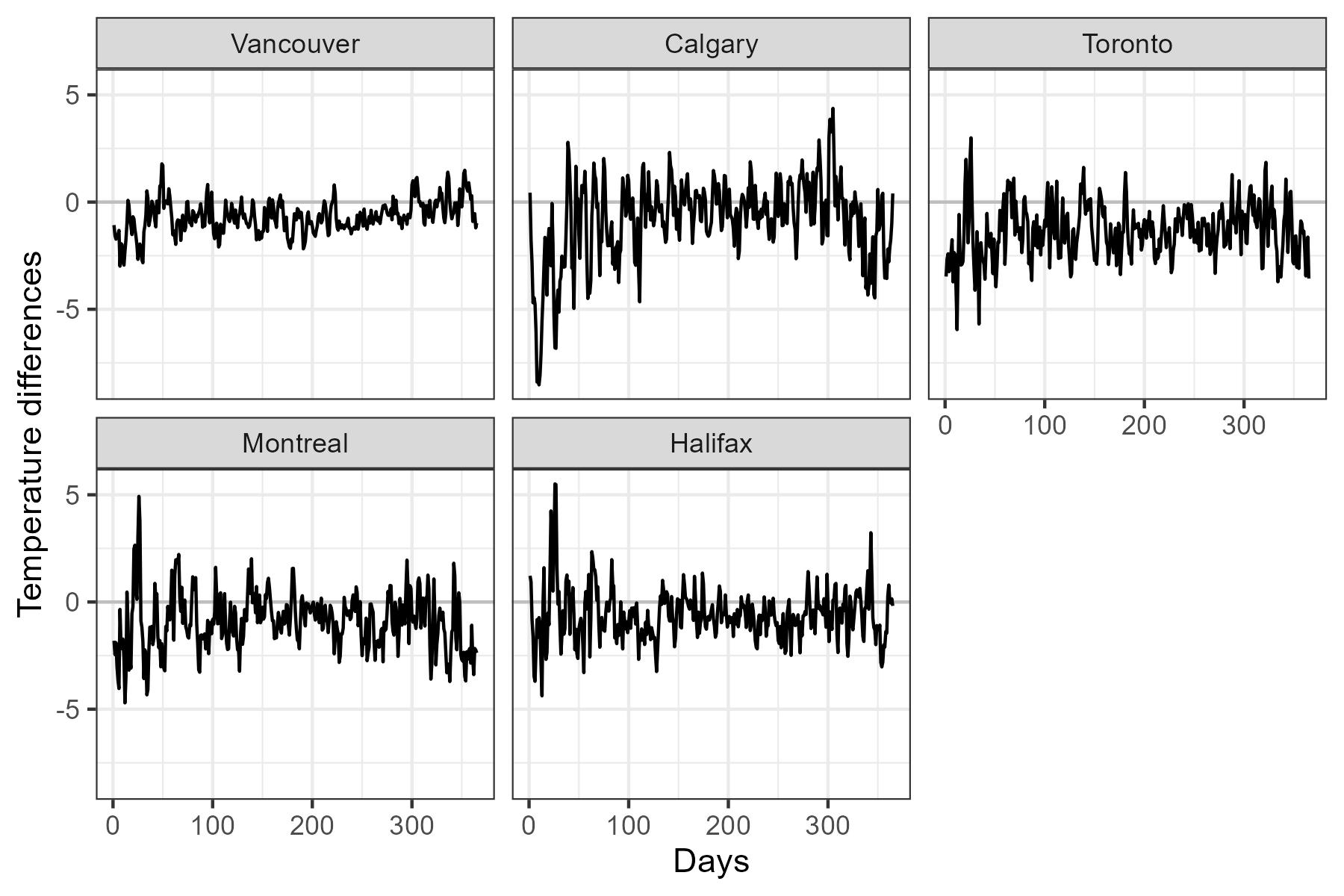}
	\caption{
		The differences between the temperature curves averaged over two groups of years 1962--1981 and 1982--2011. 
		The horizontal line at zero is colored in gray.
	}
	\label{fig_rda_2test_CW}
\end{figure}

The differences $\bar{X}_1-\bar{X}_2$ in mean curves for each city are represented in \autoref{fig_rda_2test_CW}. 
These curves show clear difference from zero particularly in shape, 
while might be challenging to detect using classical approaches.
We conduct the bootstrap two-sample functional mean tests with statistics as considered in \autoref{ssec_5_1}.
In addition to the depth statistics, the analysis considers only our main competitors: QuadInv, QuadMul, L2, sup, Fint, and Fmax. 
This is because the sampling distributions of these statistics are approximated using the same resampling approach, allowing for a clearer comparison of the differences attributable solely to the test statistics.
We select the tuning parameter as implemented in \autoref{ssec_5_1}, when needed.

The p-values from the bootstrap tests under comparison are summarized in
\autoref{tb_rda_2test_pval}.
Notably, the KD statistic produces a consistent result,
 which all explain  significant differences in mean,
while the other depth statistics sometimes provide weaker evidence depending on cities. 
Among the main competing tests, the $L^2$ statistic the most consistently rejects the null hypothesis,
while the other tests sometimes exhibit very high p-values.

\begin{table}[h!]\small
	\renewcommand{\arraystretch}{1.3}
	\centering
	\caption{
		The p-values from the bootstrap two-sample functional mean tests. 
		We provide the results from the proposed depth statistics and our main competitors considered in \autoref{ssec_5_1}. 
		Significant p-values no bigger than the level 0.05 are indicated in \textit{italic}.
	}
	\label{tb_rda_2test_pval}
	\medskip
	\resizebox{\textwidth}{!}{\begin{tabular}{c|cccccccccc} \hline
		& ITD   & IFD   & RHD   & KD & QuadInv & QuadMul & L2 & sup & Fint  & Fmax  \\ \hline
		Vancouver  & \textit{0.003} & \textit{0.000}     & \textit{0.000}     & \textit{0.004}  & \textit{0.028}   & 0.143   & \textit{0.007}    & 0.238   & \textit{0.002} & \textit{0.047} \\
		Calgary  & 0.216 & \textit{0.000}     & \textit{0.033} & \textit{0.050}   & 0.269   & 0.304   & 0.063    & 0.076   & 0.244 & 0.305 \\
		Toronto  & \textit{0.000}     & \textit{0.000}     & \textit{0.000}     & \textit{0.001}  & \textit{0.004}   & 0.080    & \textit{0.001}    & \textit{0.040}    & \textit{0.001} & 0.185 \\
		Montreal & \textit{0.013} & 0.273 & 0.778 & \textit{0.028}  & 0.168   & 0.583   & \textit{0.039}    & 0.395   & 0.084 & 0.512 \\
		Halifax & \textit{0.034} & \textit{0.000}     & 0.937 & \textit{0.048}  & 0.558   & 0.913   & 0.060     & 0.059   & 0.075 & 0.291 \\ \hline
	\end{tabular}}
\end{table}

\subsection{Regression analysis of electricity prices on wind power in-feed}

We illustrate the proposed depth statistic in FoFR framework in this section.
We are interested in the relationship between electricity prices $Y_i$ and the wind power in-feed (kWh) $X_i$ in Germany.
Both are treated as functions of time within a day, and the time was scaled into the unit interval [0,1] before the analysis.
This dataset was analyzed by \cite{Liebl13}, where the original dataset can be found.
We use the version with sample size $n=591$, 
preprocessed by \cite{CMR25} for prediction in FoFR models.
The observed regressor and response curves are presented with gray color in \autoref{fig_fofr_rda_EW} with 20 randomly selected curves in black.

\begin{figure}[h!]
	\centering
	\includegraphics[width=0.99\linewidth]{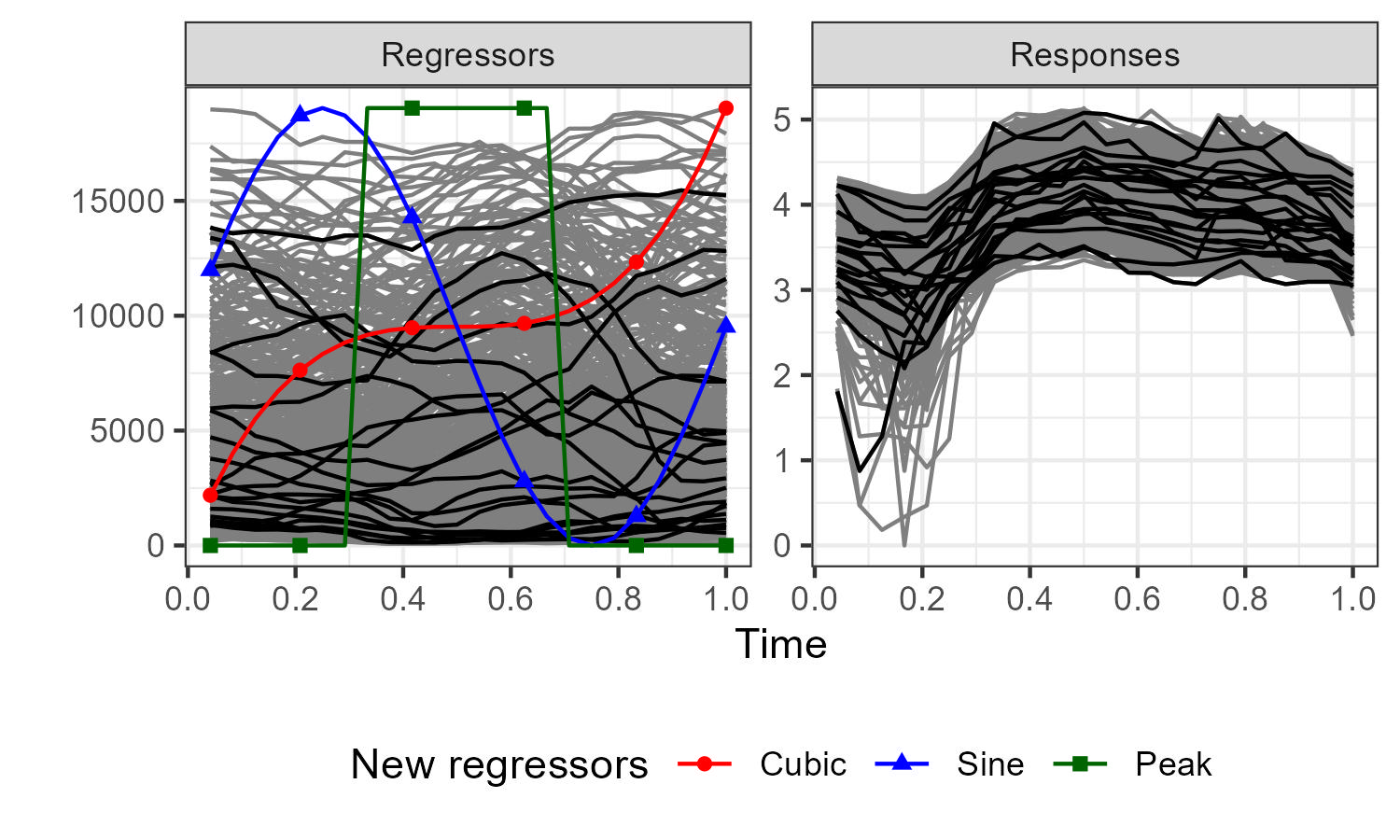}
	\caption{
		Regressor (wind power in-feed, kWh) and response (electricity prices) curves in gray.
		Twenty randomly selected curves are marked in black.
		Three artificial new regressors $\{X_{0,l}: l \in \{\rm cub, sin, peak\}\}$ are indicated with red circles (cubic), blue triangles (sine), and green squares (peak). 
	}
	\label{fig_fofr_rda_EW}
\end{figure}

For illustration of the proposed method, we consider artificial new regressors resembling cubic, sine, and peak curves, scaled by $X_{\max} \equiv \max_{1 \leq i \leq n} \max_{0 \leq t \leq 1} X_i(t)$:
\begin{align*}
	X_{0,\mathrm{cub}}(t) & = X_{\max}\{4(t-1/2)^3+1/2\}, \\
	X_{0,\sin}(t) & = X_{\max}\{\sin(2\pi t)/2+1/2\}, \\
	X_{0,\mathrm{peak}}(t) & = X_{\max}\I(0.3 \leq t \leq 0.7).
\end{align*}
The corresponding unobserved response curves are then expected to differ from the global mean response curve $\bar{Y} \equiv n^{-1} \sum_{i=1}^n Y_i$, due to their distinct patterns,
thereby increasing the complexity of data analysis. 
These new regressor curves $\{X_{0,l}: l \in \{\rm cub, sin, peak\}\}$ are shown in different colors and points in the left panel of \autoref{fig_fofr_rda_EW}.

\autoref{tb_fofr_rda_EW} shows the p-values from the bootstrap tests with the statistics considered in \autoref{ssec_5_2}.
The relevant tuning parameters are chosen as in \autoref{ssec_5_2}.
The depth statistics with RHD and KD consistently yield significant results,
whereas the $L^2$ and supremum norm statistics provide weaker evidence.
In particular, each of the latter two exhibits complementary strengths:
the $L^2$ statistic is effective in detecting deviations of the new response at cubic and sine curves from the global mean response 
but less sensitive for the peak curve,
while the supremum statistic shows the opposite pattern.
In contrast, the proposed depth statistics produce robust and consistent results across all new regressors.

\begin{table}[h!]
	\renewcommand{\arraystretch}{1.3}
	\centering
	\caption{
		The p-values of the bootstrap tests with the statistics consider in \autoref{ssec_5_2} for FoFR analysis of electricity prices on wind power in-feed. 
		Significant p-values no bigger than the level 0.05 are indicated in \textit{italic}.
		}
	\label{tb_fofr_rda_EW}
	\medskip
	\begin{tabular}{c|cccc} \hline
		& RHD   & KD    & L2    & sup   \\  \hline
		Cubic  & \textit{0.031} & \textit{0.037} & 0.089 & 0.682 \\
		Sine  & \textit{0.000}     & \textit{0.000}  & 0.053 & 0.159 \\
		Peak & \textit{0.000} & \textit{0.000}  & 0.140  & \textit{0.011} \\ \hline
	\end{tabular}
\end{table}

\section{Concluding remarks and outlook}

In this paper, we propose a depth statistic for inference on functional parameters. 
The proposed approach is particularly useful when scientifically meaningful alternatives are complex in shape and therefore not easily captured by existing methods.
We establish the theoretical validity of the procedure, demonstrate its empirical performance through comprehensive numerical studies under a wide range of alternatives, and further illustrate its practical value with two real-data examples. 
Taken together, these results show that depth-based inference provides a useful and effective tool for FDA, 
especially in problems where shape-driven departures from the null are of primary interest.

This work also opens several avenues for future research on inference for complex data, where depth notions can play an important role. 
One promising direction is to investigate connections with confidence distributions \citep{XS13}, which may offer a broader route to inference under composite hypotheses,
particularly since infinite-dimensional settings generally do not admit a density in the usual sense \cite[cf.][]{DH10}. 
It would also be of interest to extend our framework to broader classes of object data \citep{SZZM26} by building on recent developments in data depth for random objects \cite[e.g.,][]{DL23}. 
More broadly, we hope that the present work provides a useful starting point for further depth-based methodology in modern inference problems involving complex data.

%\tred{
%While not stated explicitly, the current method assumes that the null hypothesis is simple.
%This can be seen in formula (1), where Pn is the distribution of the test statistics under
%the null hypothesis. For a composite hypothesis, there is in general no unique distribution
%under the null, and thus the present approach is not applicable. This severely limits the
%present methods.
%}
%
%extension: confidence distribution, composite hypotheses
%
%1) we propose depth statistic for functional data inference problems
%2) It is particuarly useful and powerful when alternative functional parameters in scientific questions are complicated in their shapes,
%which are limited from existing methods. 
%3) we demonstrate our method through comprehensive numerical studies with various types of alternatives and showcase through two real data examples.
%
%1) This research opens new avenues related to complex data inference, where depth are useful.
%2) confidence distribution \cite{XS13}, as no density for infinite-dimensional data \cite{DH10}
%3) extend it to \cite{SZZM26}, with object depth \cite{DL23}

\section*{Acknowledgement}

The author thanks Stanislav Nagy for helpful discussion about Conditions~\ref{condContPoint}--\ref{condContProb}.

%\section*{Disclosure statement}
%
%The authors report there are not competing interests to declare.
%
%
%\section*{Data availability statement}
%
%The weather and wind power datasets analyzed in this paper can be obtained through the R package \texttt{weathercan} \citep{weathercan} and the work by \cite{CMR25}, respectively.

%\newpage
\bibliographystyle{dcu}
\bibliography{DepthInfer}

\newpage

\setcounter{page}{1}

\normalsize
\renewcommand{\baselinestretch}{1.0}

\appendix

\centerline{\Large \bf Supplementary Material}

\renewcommand\theequation{\thesection.\arabic{equation}}
\newcounter{daggerfootnote}
\newcommand*{\daggerfootnote}[1]{%
	\setcounter{daggerfootnote}{\value{footnote}}%
	\renewcommand*{\thefootnote}{\fnsymbol{footnote}}%
	\footnote[2]{#1}%
	\setcounter{footnote}{\value{daggerfootnote}}%
	\renewcommand*{\thefootnote}{\arabic{footnote}}%
}
%\begin{center}
%	\large Hyemin Yeon\daggerfootnote{hyeon2@gmu.edu}
%	
%	\medskip
%	Department of Statistics, 
%	George Mason University
%%	Kent, OH, 44242, USA.
%	
%\end{center}
\numberwithin{thm}{section}
\numberwithin{eg}{section}
\numberwithin{defn}{section}
\numberwithin{prop}{section}
\numberwithin{lem}{section}

\bigskip\noindent
This is a supplementary document of the manuscript titled ``Effective and flexible depth-based inference for functional parameters''.
This supplement includes technical and numerical details in Sections~\ref{AppTheo}--\ref{AppNumDetails}, respectively.

\section{Theoretical details} \label{AppTheo}

\subsection{Proof of the main results}\label{AppProofMain}

\begin{proof}[Proof of \autoref{thmDepthStat}]
	Let $\e,\dt>0$.
	Then, since $T_n$ is tight, there exists a compact set $G \subseteq \HH$ such that $\sup_{n \in \N} \pr(T_n \notin G) < \dt$. 
	Then, with  $A_n \equiv D(T_n;P_n) - D(T_n;P)$, 
	\begin{align*}
		\pr(|A_n|>\e)
		\leq \pr \left( \sup_{x \in G} |D(x;P_n) - D(x;P)| > \e \right)
		+ \pr(T_n \notin G),
	\end{align*}
	implying that $\limsup_{n\to\infty} \pr(|A_n|>\e) \leq \dt$.
	This yields $A_n \xrightarrow{\pr} 0$.
%	Write $A_n \equiv D(T_n;P_n) - D(T_n;P)$ so that
%	\begin{align*}
%		|A_n|
%		\leq \sup_{x \in \HH} |D(x;P_n) - D(x;P)| \to 0. 
%	\end{align*}
	Note that this convergence holds due to Condition~\ref{condConvDistnData} and~\ref{condContProb}.
	Then, by Condition~\ref{condConvDistnData} and~\ref{condContPoint},
	we further observe
	\begin{align*}
		D(T_n;P_n) = D(T_n;P) + A_n \xrightarrow{\mathsf{d}} D(T;P).
	\end{align*}
	By the Polya's theorem \citep[Theorem~9.1.4]{AL06} and Condition~\ref{condLimitCont},
	we have
	\begin{align*}
		\sup_{u \in \R}|\pr(D(T_n;P_n) \leq u) - \pr(D(T;P) \leq u)| \to 0.
	\end{align*}
	Similarly, 
	denoting $A_n^* \equiv D(T_n^*;\widehat{P}_n) - D(T_n^*;P)$,
	we can find that $A_n^* \xrightarrow{\pr^*} 0$,
%	\begin{align*}
%		\pr^*(|A_n^*|>\e)
%		\leq \pr^* \left( \sup_{x \in G} |D(x;\widehat{P}_n) - D(x;P)| > \e \right)
%		+ \pr^*(T_n^* \notin G),
%	\end{align*}
%	\begin{align*}
%		|A_n^*|
%		\leq \sup_{x \in \HH} |D(x;\widehat{P}_n) - D(x;P)| \xrightarrow{\pr} 0,
%	\end{align*}
	using Conditions~\ref{condConvDistnData}, \ref{condConvDistnBTS}, and~\ref{condContProb}.
	This is derived by using a subsequence argument \cite[e.g.,][Theorem~20.5]{bill95}.
	Then, using Conditions~\ref{condContPoint} and~\ref{condLimitCont}, and the Polya's theorem \citep[Theorem~9.1.4]{AL06},
	we have (possibly following a subsequence) that
	\begin{align*}
		\sup_{u \in \R}|\pr^*(D(T_n^*;\widehat{P}_n) \leq u) - \pr(D(T;P) \leq u)| \xrightarrow{\pr} 0.
	\end{align*}
	This concludes the desired result.
\end{proof}

\begin{proof}[Proof of \autoref{thmDepthPval}]
	First, by the proof of \autoref{thmDepthStat}, we have 
	\begin{align*}
		|\widehat{F}_n(D(T_n; \widehat{P}_n)) - F(D(T_n;\widehat{P}_n))|
		\leq \sup_{u \in \R} |\widehat{F}_n(u) - F(u)| \xrightarrow{\pr} 0
	\end{align*}
	Conditions~\ref{condConvDistnData}-\ref{condConvDistnBTS} and~\ref{condContProb}
	along with the argument in the proof of \autoref{thmDepthStat}
	implies $|D(T_n;\widehat{P}_n)-D(T_n;P) |\xrightarrow{\pr} 0$, 
	and 
	Conditions~\ref{condConvDistnData} and~\ref{condContPoint} along with continuous mapping theorem yields $|D(T_n;P) - D(T;P)| \xrightarrow{\pr} 0$;
	these together imply $D(T_n;\widehat{P}_n) \xrightarrow{\pr} D(T;P)$. 
%	
%	we also observe that
%	\begin{align*}
%		|D(T_n;\widehat{P}_n)-D(T_n;P) |
%%		& \leq \sup_{x \in \HH} |D(x;\widehat{P}_n) - D(x;P)| \to 0, \\
%		& \xrightarrow{\pr} 0, \\
%		|D(T_n;P) - D(T;P)| 
%		& % \leq \sup_{P \in \pp(\HH)} |D(T_n;P)-D(T;P)|
%		\xrightarrow{\pr} 0,
%	\end{align*}
%	and Conditions~\ref{condConvDistnData} and~\ref{condContPoint},
%	
%	implying $D(T_n;\widehat{P}_n) \xrightarrow{\pr} D(T;P)$. 
%	Here, deriving the second result may require a subsequence argument \cite[e.g.,][Theorem~20.5]{bill95}.
	The probability integral transform along with Slutsky's theorem \cite[cf.][Theorem~3.1]{bill99} concludes
	\begin{align*}
		\widehat{\pi}_n
		&= \widehat{F}_n(D(T_n; \widehat{P}_n)) - F(D(T_n;\widehat{P}_n))
		+ F(D(T_n;\widehat{P}_n))
		\xrightarrow{\mathsf{d}} F(D(T;P)) \sim \mathsf{Unif}(0,1).
	\end{align*}
\end{proof}

\subsection{Proof for the two-sample functional mean tests} \label{App2test}

\begin{proof}[Proof of \autoref{thm_two}]
	By a central limit theorem in a general Hilbert space  \cite[e.g.,][Theorem~7.7.6]{HE15},
	we have
	\begin{align*}
		\sqrt{n_1n_2 \over n_1+n_2} (\bar{X}_1-\mu_1)
		& \xrightarrow{\mathsf{d}} \mathsf{G}(0,(1-\zeta)\ga_1), \\
		\sqrt{n_1n_2 \over n_1+n_2} (\bar{X}_2-\mu_2)
		& \xrightarrow{\mathsf{d}} \mathsf{G}(0,\zeta\ga_2), 
	\end{align*}
	implying  
	\begin{align*}
		T_{\mathrm{two},\bm{n}} 
		= \sqrt{n_1n_2 \over n_1+n_2} (\bar{X}_1 - \bar{X}_2)
		\xrightarrow{\mathsf{d}} \mathsf{G}(0,(1-\zeta)\ga_1 + \zeta\ga_2).
	\end{align*}
	Here, $ \mathsf{G}(0,(1-\zeta)\ga_1 + \zeta\ga_2)$ stands for the Gaussian distribution on $\HH$ with mean zero and covariance $(1-\zeta)\ga_1 + \zeta\ga_2$.
	This derives the first result.
	
	For the second part, let $E_k \equiv \left( \widehat{\ga}_k \to \ga_k \right)$ for $k=1,2$ and $E \equiv E_1 \cap E_2$ so that $\pr(E)=1$ by the law of large numbers \citep[Theorem~7.7.5]{HE15}.
	Then, on $E$, we have the following observations.
	Write $\bar{\e}_k^* \equiv n_k^{-1} \sum_{i=1}^{n_k} \e_{ki}^*$.  
	Using $\eo^*[\e_{ki}^*] = 0$, $\var^*[\e_{ki}^*] = \eo^*[(\e_{ki}^*)^{\otimes 2}] = \widehat{\ga}_k$, and the independence between $\{\e_{ki}^*\}_{i=1}^{n_k}$, 
	we can see that $\sqrt{n_k}\bar{\e}_k^*$ is tight with respect to $\pr^*$,
	because
	\begin{align*}
		\eo^*[\|\sqrt{n_k}\bar{\e}_k^*\|^2]
		\leq \eo^*[\|\e_k^*\|^2]
		= \mathrm{tr}(\var^*[\e_k^*])
		= \mathrm{tr} (\widehat{\ga}_k)
		\to \mathrm{tr} (\ga_k),
	\end{align*}
	implying
	\begin{align*}
		\limsup_{n\to\infty}\pr(\|\sqrt{n_k}\bar{\e}_k^*\| > M) \leq M^{-2} \mathrm{tr} (\ga_k) \to 0, \quad \text{as} \quad M\to\infty.
	\end{align*}
	Also, let $v \in \HH$ so that $\eo[\langle \e_{ki}^*, v \rangle] = 0$ and $\var^*[\langle \e_{ki}^*, v \rangle] = \langle \widehat{\ga}_k v, v \rangle \to \langle \ga_k v, v \rangle$.
	Then, following $\pr^*$, we have 
	\begin{align*}
		\langle \ga_k v, v \rangle^{-1/2} \langle \sqrt{n_k} \bar{\e}_k^*, v \rangle
		= \left( \langle \widehat{\ga}_k v, v \rangle \over \langle \ga_k v, v \rangle \right)^{1/2}  \sqrt{n_k} n_k^{-1} \sum_{i=1}^{n_k} \langle \widehat{\ga}_k v, v \rangle^{-1/2} \langle \e_{ki}^*, v \rangle
		\xrightarrow{\mathsf{d}} \nd(0,1),
	\end{align*}
	i.e., $\langle \sqrt{n_k} \bar{\e}_k^*, v \rangle \xrightarrow{\mathsf{d}} \nd(0,\langle \ga_k v, v \rangle)$.
	In brief, on $E$, \cite[Theorem~2.3]{Bosq00} implies that $\sqrt{n_k} \bar{\e}_k^*$ converges in distribution to $\mathsf{G}(0,\ga_k)$ for each $k=1,2$,
	implying that 
	\begin{align*}
		T_{\mathrm{two},\bm{n}}^* 
		= \sqrt{n_1n_2 \over n_1+n_2} (\bar{\e}_1^* - \bar{\e}_2^*)
		\xrightarrow{\mathsf{d}} \mathsf{G}(0,(1-\zeta)\ga_1 + \zeta\ga_2).
	\end{align*}
	Since $\pr(E)=1$, this concludes part~(b). 
\end{proof}

\subsection{Technical details for the kernel depth} \label{AppContKD}

We provide the results about the KD that are more stronger than Condition~\ref{condContPoint}--\ref{condContProb},
where their proofs are given in \autoref{AppContKDsub}. 

\begin{enumerate}[({K}1)]
	\item \label{condK1}
	The function $K:[0,\infty)\to[0,\infty)$ is a continuous and non-increasing function with $K(0)>0$ and $\lim_{t\to\infty} K(t) = 0$
	\item \label{condK2}
	The function $K$ is positive definite in the sense that
	\begin{align*}
		\sum_{1 \leq j,j' \leq m} a_ja_{j'}K(\|x_j-x_{j'}\|) \geq 0,
		\quad \forall m \in \N, 
		\quad \forall a_1, \dots, a_m \in \R, 
		\quad \forall x_1, \dots, x_m \in \HH. 
	\end{align*}
\end{enumerate}

\begin{thm} \label{thmKDcont}
	%	Let $K:[0,\infty)\to\R$ be a bounded and continuous function. 
	Suppose that the kernel function $K$ satisfies Condition~\ref{condK1}--\ref{condK2}. 
	Regarding the KD $D_K(x;Q)$ for $x \in \HH$ and $Q \in \pp(\HH)$ defined in Equation~\eqref{eqKDpop},
	we have the following two continuity properties.
	\begin{enumerate}[(a)]
		\item 
		The map $x \mapsto D_K(x;Q)$ is continuous uniformly over $Q \in \pp(\HH)$ in the sense that, 
		for a sequence $\{x_n\}$ in $\HH$ and $x \in \HH$, as $n\to\infty$, if $x_n\to x$, then we have
		\begin{align*}
			\sup_{Q \in \pp(\HH)}|D_K(x_n;Q) - D_K(x;Q)| \to 0.
		\end{align*}
		\item 
		The map $Q \mapsto D_K(\cdot; Q)$ is continuous uniformly over $x \in \HH$ in the sense that, 
		for a sequence $\{Q_n\}$ in $\pp(\HH)$ and $Q \in \pp(\HH)$,
		as $n\to\infty$, if $Q_n$ weakly converges to $Q$,
		then we have
		\begin{align*}
			\sup_{x \in \HH}|D_K(x;Q_n) - D_K(x;Q)| \to 0.
		\end{align*}
	\end{enumerate}
	
\end{thm}

\begin{eg} \label{egKDlimitCont}
	Let $P \in \pp(\HH)$ be the Gaussian with mean element $0 \in \HH$ and a covariance operator $\sa:\HH\to\HH$.
	Then, when using the Gaussian kernel given by $K(t) = e^{-t^2/2}$ for $t \in [0,\infty)$,
	the kernel depth of $x \in \HH$ with respect to $P$ is explicitly written by \citet[Example~1]{WN25} as
	\begin{align*}
		D_K(x;P)
		= \left\{ \prod_{j=1}^\infty (1+\s_j) \right\}^{-1/2} \exp \left[ - {1 \over 2} \sum_{j=1}^\infty {\langle x, \psi_j \rangle^2 \over 1+\s_j} \right],
	\end{align*}
	where $(\s_j,\psi_j)$ denotes the $j$-th eigenvalue-eigenfunction pair of $\sa$. 
	Let $T \sim P$, which is the case in most asymptotic distributional theories.
	Then, since $\{\s_j^{-1/2}\langle T, \psi_j \rangle\}_{j=1}^\infty$ is a sequence of independent standard normal variables,
	the random variable
	\begin{align*}
		V \equiv \sum_{j=1}^\infty {\langle T, \psi_j \rangle^2 \over 1+\s_j}
		\ed \sum_{j=1}^\infty {\s_j \over 1+\s_j} (\s_j^{-1/2}\langle T, \psi_j \rangle)^2
	\end{align*}
	is a infinite mixture of independent $\chi^2$ random variables with one degree of freedom,
	which is hence continuous.
	Note that $\pr(0<V<\infty)=1$, implying that $\pr(D_K(T;P) > 0)=1$.
	Since the map $\R \ni a \mapsto e^{-a/2} \in \R$ is one-to-one, $D_K(T;P)$ is a continuous random variable.
\end{eg}

\subsubsection{Proof of \autoref{thmKDcont}} \label{AppContKDsub}

%Write $\pp(\HH)$ for the set of all probability measures on $\HH$. 
We start with a series of definitions related to the kernel theory.
We refer to \cite{WN25} for a recent theory for kernel-based methods in Statistics and Machine Learning.

\begin{defn}
	We call a function $\ka:\HH\times\HH\to\R$ a \emph{kernel} on $\HH$ if
	\begin{enumerate}[(a)]
		\item $\ka$ is symmetric as $\ka(x,y) = k(y,x)$ for all $x\in\HH$ and $y \in\HH$, and
		
		\item $\ka$ is positive definite, i.e., $\sum_{j=1}^p \sum_{j'=1}^p \ka(x_j, x_{j'}) \geq 0$ for all $(x_1, \dots, x_p) \in \HH^p$, $(a_1, \dots, a_p) \in \R^p$, and $p \in \N$. 
		
	\end{enumerate}
\end{defn}

\begin{defn}
	For a kernel $\ka$, 
	a Hilbert space $\G(\ka)$ of functions from $\HH$ to $\R$ is called a \emph{reproducing kernel Hilbert space (RKHS)} if
	\begin{enumerate}[(a)]
		\item $\ka(\cdot, x) \in \G(\ka)$ for all $x \in \HH$, and
		
		\item (reproducing property) $\langle g, \ka(\cdot, x) \rangle_\ka = g(x)$ for all $g \in \G(\ka)$ and $x \in \HH$. 
		
	\end{enumerate}
\end{defn}

\begin{defn}
	For a kernel $\ka$, we define the \emph{kernel mean embedding (KME)} $\ee_\ka:\pp(\HH) \to \G(\ka)$ as $\ee_\ka(P) = \int \ka(\cdot, x) dP(x)$ for $P \in \pp(\HH)$, i.e.,
	\begin{align*}
		\{\ee_\ka(P)\}(x) = \int \ka(x, y) dP(y), \quad x \in \HH, \quad P \in \pp(\HH).
	\end{align*}
	%	\cite[Definition~4]{WN25}.
\end{defn}
We observe that, 
for $g \in \G(\ka)$ and $P \in \pp(\HH)$,
\begin{align*}
	\langle \ee_\ka P, g \rangle_\ka
	= \int \langle \ka(\cdot, x), g \rangle dP(x)
	= \int g(x) d P(x)
	= \int g dP ,
	%	= Pg
\end{align*} 
and hence, $\langle \ee_\ka P, \ka(\cdot, x) \rangle_\ka = \{\ee_\ka (P)\}(x)$ for all $x \in \HH$. 
%\cite[cf.][Equation~(11)]{WN25}.

With a slight abuse of notation, 
we define an inner product between $P$ and $Q$ with respect to $\ka$ as $\langle P, Q \rangle_\ka \equiv \langle \ee_\ka P, \ee_\ka Q \rangle_\ka$ for $P\in \pp(\HH)$ and $Q \in \pp(\HH)$. 
When $\ka$ is Gaussian, i.e., $\ka(x,y) = e^{-\|x-y\|^2/2}$, $x\in\HH$, $y\in\HH$, it indeed induces an inner product on $\pp(\HH)$ as $\ka$ is characteristic by \citet[Theorem~11]{WN25}.
Note that, for $P\in \pp(\HH)$ and $Q \in \pp(\HH)$,
\begin{align*}
	\langle P, Q \rangle_\ka = \langle \ee_\ka P, \ee_\ka Q \rangle_\ka
	= \int \int \ka(x,y) dP(x) dQ(y)
	= (P \otimes Q)(\ka),
\end{align*}
where $P \otimes Q$ represents the product measure of $P$ and $Q$ \cite[cf.][Section~2]{bill99}.
In particular, the corresponding distance $\|P-Q\|_\ka \equiv \langle P-Q, P-Q \rangle_\ka^{1/2}$ between $P$ and $Q$ is called the \emph{maximum mean discrepancy (MMD)} \cite[Definition~5]{WN25}. 

The following lemma states that the weak convergence of probability measures implies the convergence in MMD. 
\begin{lem} \label{lemMMDcharWeakConv}
	Let $\ka$ be a bounded and continuous kernel on $\HH$. 
	For a sequence $\{P_n\}$ in $\pp(\HH)$ and $P \in \pp(\HH)$,
	as $n\to\infty$, if $P_n$ weakly converges to $P$,
	then $\|P_n-P\|_k \to 0$. 
\end{lem}
\begin{proof}
	We follow the argument in the proof of \citet[Lemma~5(iv)$\Rightarrow$(i)]{SBSM23}.
	Due to the weak convergence of $\{P_n\}$ to $P$,
	\citet[Theorem~2.8(ii)]{bill99} yields that $P_n \otimes P$, $P \otimes P_n$, and $P_n \otimes P_n$ all weakly converge to $P \otimes P$.
	Since $\ka$ is also bounded and continuous,
	we have
	\begin{align*}
		\|P_n-P\|_\ka^2
		& = \langle P_n, P_n \rangle_\ka - \langle P_n, P \rangle_\ka - \langle P, P_n \rangle_\ka + \langle P, P \rangle_\ka
		\\& = (P_n \otimes P_n)(\ka) - (P_n \otimes P)(\ka) - (P \otimes P_n)(\ka) + (P \otimes P)(\ka)
		\to 0.
	\end{align*}
\end{proof}

\begin{proof}[Proof of \autoref{thmKDcont}]
	
	We define $\ka:\HH\times\HH\to\R$  by $\ka(x,y) \equiv K(\|x-y\|)$ for $x\in\HH$ and $y\in\HH$.
	%	 is a bounded and continuous kernel on $\HH$.
	Due to the equivalence between KME (induced by $\ka$) and the KD (determined by $K$) in \citet[Theorem~1]{WN25}, it suffice to show the convergence upon embedding the probability measures into $\G(\ka)$.
	
	The first part is derived as follows.
	Using the symmetry of the kernel $\ka$ and its reproducing property,
	we find that, for each $x \in \HH$ and $y \in \HH$, 
	\begin{align*}
		\|\ka(\cdot, x)-\ka(\cdot, y)\|_\ka^2
		= \ka(x,x) + \ka(y,y) - 2\ka(x,y).
	\end{align*}
	This implies that $\|\ka(\cdot, x_n)-\ka(\cdot, x)\|_\ka \to 0$ as $x_n\to x$ due to the continuity of $\ka$. 
	On the other hand,
	the boundedness of $\ka$ yields
	\begin{align*}
		\sup_{P \in \pp(\HH)} \|\ee_\ka P\|_\ka^2
		= \sup_{P \in \pp(\HH)} \int \int \ka(x,y) dP(x) dP(x)
		\leq \sup_{x \in \HH} \ka(x,x) < \infty.
	\end{align*}
	Combining these, we conclude that 
	\begin{align*}
		\sup_{P \in \pp(\HH)} |\ee_\ka P(x_n) - \ee_\ka P(x)|
		& = \sup_{P \in \pp(\HH)} |\langle \ee_\ka P, \ka(\cdot, x_n) \rangle_\ka - \langle \ee_\ka P, \ka(\cdot, x) \rangle_\ka|
		\\& \leq \|\ka(\cdot, x_n)-\ka(\cdot, x)\|_\ka \sup_{P \in \pp(\HH)} \|\ee_\ka P\|_\ka \to 0. 
	\end{align*}
	
	Next, we prove part~(b).
	Following the proof of \citet[Theorem~5]{WN25}, 
	since $\|\ka(\cdot, x)\|_\ka^2 = \ka(x,x)$ for $x \in \HH$, 
	it holds as
	\begin{align*}
		\sup_{x \in \HH} |\ee_\ka P_n(x) - \ee_\ka P(x)|
		& = \sup_{x \in \HH} |\langle \ee_\ka P_n, \ka(\cdot, x) \rangle_\ka - \langle \ee_\ka P, \ka(\cdot, x) \rangle_\ka|
		\\& \leq \sup_{x \in \HH} \|\ka(\cdot, x)\|_\ka \|\ee_\ka P_n-\ee_\ka P\|_\ka
		\\& = \sup_{x \in \HH} \sqrt{\ka(x,x)} \|P_n-P\|_\ka \to 0,
	\end{align*}
	where the last convergence follows by \autoref{lemMMDcharWeakConv}.
\end{proof}

%\begin{proof}
%	kernel depth is continuous in P?
%	
%	\cite[Theorem~11]{WN25}, Gaussian kernel is characteristic to $\pp(\HH)$,
%	
%	\cite[Theorem~8]{SS18}, it is characteristic to 
%	
%	$\mm$ \citep{SBSM23} the set of finite Radon (signed and regular) measures (= $\mm_f$ in \cite{SS18})
%	
%	$\mm_f^0 = \{ \mu \in \mm_f: \mu(\X) = 0 \}$
%	
%	
%	
%	
%	
%	
%	
%\end{proof}

\subsection{Nondegeneracy of the modified RHD} \label{AppRHDnondege}

We provide a nondegeneracy result for the modified RHD given in  \autoref{ssecFdepthRHD} under a mild condition, which includes all elliptically-contoured functional data \citep{BBT14}.

\begin{thm}
	Let $Q \in \pp(\HH)$ and we write 
	\begin{align*}
		\vv_\ld(Q) \equiv \{v \in \HH:\|v\|=1, \|\ga(Q)^{1/2}v\| \geq \ld\}.
	\end{align*}
	Suppose that 
	\begin{align} \label{eqRHDcondNonDege}
		\sup_{v \in \vv_\ld(Q)}\pr \left( {\langle X-\eo[X], v \rangle \over \|\ga(Q)^{1/2}v\|} \leq u \right) < 1, \quad \forall u \in \R.
	\end{align}
	Then, $D_{\mathrm{RHD},\ld}(x;Q)>0$ for all $x \in \HH$. 
\end{thm}
\begin{proof}
	This holds because
	\begin{align*}
		D_{\mathrm{RHD},\ld}(x;Q)
		& = \inf_{v \in \vv_\ld(Q)} \pr(\langle X, v \rangle \geq \langle x, v \rangle)
		\\& = 1 - \sup_{v \in \vv_\ld(Q)}\pr \left( {\langle X-\eo[X], v \rangle \over \|\ga(Q)^{1/2}v\|} \leq {\langle x-\eo[X], v \rangle \over \|\ga(Q)^{1/2}v\|} \right)
		\\& \geq 1 - \sup_{v \in \vv_\ld(Q)}\pr \left( {\langle X-\eo[X], v \rangle \over \|\ga(Q)^{1/2}v\|} \leq \ld \|x-\eo[X]\| \right)
		>0.
	\end{align*}
\end{proof}

\subsection{Uniform continuity of the modified RHD in measure}
\label{appRHDconsist}

For $Q \in \pp(\HH)$ with positive-definite covariance operator $\ga(Q)$
and each $\ld \in (0,\|\ga(Q)^{1/2}\|)$, % with $\vv_\ld  = \vv_\ld(P)$,
we define a function $M_\ld(\cdot;Q):\HH \to \R$  by
\begin{align} \label{eq_RHDinMft}
	M_\ld(f;Q) \equiv \sup_{v \in \vv_\ld(Q)}{|\langle f, v \rangle| \over \|\ga(Q)^{1/2}v\|}, \quad f \in \HH,
\end{align}
where $\vv_\ld(Q) \equiv \{ v \in \HH:\|v\|=1, \|\ga(Q)^{1/2}v\| \geq \ld \}$. 
We now suppose that $P$ is a Gaussian distribution on $\HH$ with covariance operator $\ga= \ga(P)$ being a positive-definite compact operator on $\HH$.
Let $\ld \in (0,\|\ga^{1/2}\|)$ and write $\vv_\ld = \vv_\ld(P)$, $M_\ld = M_\ld(\cdot;P) $, and $D_\ld = D_{\mathrm{RHD},\ld}$. 
Then, the modified RHD in \eqref{eqRHDmod} is reduced to
\begin{align*}
	D_\ld(x;P) 
	& = \inf_{v \in \vv_\ld} \left\{ 1-\Phi \left( {\langle x-\eo[X],v \rangle \over \|\ga^{1/2}v\|} \right) \right\}
	\\& = \Phi (-M_\ld(x-\eo[X])),
\end{align*}
where $\Phi$ stands for the CDF of the standard normal distribution.
To state the consistency theorem, let $G \subseteq \HH$ be a compact subset. 
We assume this framework throughout this section.

\begin{thm} \label{thmRHDconsist}
	Under the above setup, 
	we further assume that  $\opnorm{\ga(Q_n)-\ga(P)}_\infty \to 0$ for $Q_n$ weakly converging to $P$.
%	Suppose that $P$ is a Gaussian distribution on $\HH$ with covariance operator $\ga(P)$ being a positive-definite compact operator on $\HH$. 
%	Let $\ld \in (0,\|\ga(P)^{1/2}\|)$ and $G \subseteq \HH$ be a compact subset. 
%	We also assume that $\opnorm{\ga(Q_n)-\ga(P)}_\infty \to 0$ for $Q_n$ weakly converging to $P$.
	Then, the map $Q \mapsto D(\cdot; Q)$ is continuous at $P$ uniformly over $x \in G$. In other words, 
	for any sequence $\{Q_n\}$ in $\pp(\HH)$,
	if $Q_n$ weakly converges to $P$ and $\opnorm{\ga(Q_n)-\ga(P)}_\infty \to 0$, then we have
	\begin{align*}
		\sup_{x \in G} |D_\ld(x;Q_n) - D_\ld(x;P)| \to 0. 
	\end{align*}
\end{thm}

To prove \autoref{thmRHDconsist}, we will use the following two Propositions~\ref{propRHDcpt}--\ref{propRHDcont_ld}, 
which establish uniform continuity with respect to the probability measure $P$ and the direction set $\vv_\ld$, respectively, 
with the other held fixed.

\begin{prop} \label{propRHDcpt}
	Let $H_{x,v} \equiv \{ y \in \HH: \langle y,v \rangle \geq \langle x, v \rangle \}$ denote the halfspace attached at $x \in \HH$ with normal vector $v \in \HH$.
	Also, let $\{Q_n\}$ be a sequence  in $\pp(\HH)$ weakly converging to $P \in \pp(\HH)$.
%	with $P(\{ y \in \HH: \langle y, v \rangle = c \}) = 0$ for all $v \in \HH$ and $c \in \R$. 
	Then, under the above setup, we have
	\begin{align*}
		\sup_{x \in G} \left| \inf_{v \in \vv_\ld(P)} Q_n (H_{x,v}) - \inf_{v \in \vv_\ld(P)} P (H_{x,v}) \right| \to 0.
	\end{align*}
\end{prop}

\begin{prop} \label{propRHDcont_ld}
	Under the above setup, 
	the map $\ld \mapsto D_\ld(x;P)$ is continuous uniformly over $x \in G$. 
	In other words,
	for each $\ld \in (0,\|\ga^{1/2}\|)$, as $\beta \to \ld$, we have
	\begin{align*}
		\sup_{x \in G} | D_\beta(x;P) - D_\ld(x;P) | \to 0.
	\end{align*}

\end{prop}

\begin{proof}[Proof of \autoref{thmRHDconsist}]
	By using $\opnorm{A^{1/2}-B^{1/2}}_\infty \leq \opnorm{A-B}_\infty^{1/2}$ as in \citet[Theorem~X.1.1]{Bhatia13},
	we have
	\begin{align*}
		\oa_n 
		& \equiv \sup_{v \in \HH:\|v\|=1} \| \ga(Q_n)^{1/2}v-\ga(P)^{1/2} v \|
		\leq \opnorm{\ga(Q_n)^{1/2}-\ga(P)^{1/2}}_\infty
		\\& \leq \opnorm{\ga(Q_n)-\ga(P)}_\infty 
		\to 0. 
	\end{align*}
	Then, since
	\begin{align*}
		\|\ga(Q)^{1/2}v\| \leq \|\ga(Q_n)^{1/2}v\| + \oa_n,
		\quad\text{and}\quad
		\|\ga(Q_n)^{1/2}v\| \leq \|\ga(Q)^{1/2}v\| + \oa_n,
	\end{align*}
	we obtain
	\begin{align*}
		\|\ga(P)^{1/2}v\| \geq \ld + \oa_n
		\implies \|\ga(Q_n)^{1/2}v\|  \geq \ld
		\implies \|\ga(P)^{1/2}v\| \geq \ld - \oa_n.
	\end{align*}
	If we write $H_{x,v} \equiv \{ y \in \HH:\langle y,v \rangle \geq \langle x, v \rangle \}$ for $x,v \in \HH$,
	then we derive
	\begin{align} \label{eq_thmRHDconsist}
		\inf_{v \in \vv_{\ld-\oa_n}(P)} Q_n (H_{x,v})
		\leq \inf_{v \in \vv_\ld(Q_n)} Q_n (H_{x,v})
		\leq \inf_{v \in \vv_{\ld+\oa_n}(P)} Q_n(H_{x,v}).
	\end{align}
	
	Let $\e \in (0,\infty)$ be given.
	Then, there exist $N \in \N$ and $\dt \in (0,\|\ga^{1/2}\|-\ld)$ such that
	\begin{enumerate}[(i)]
		\item for any $n \geq N$, $\oa_n <\dt$,
		\item for any $\beta\in(0,\|\ga\|^{1/2})$ with $|\beta-\ld| \leq \dt$,
		it holds that 
		\begin{align*}
			\sup_{x \in G} \left| \inf_{v \in \vv_\beta(P)} P (H_{x,v}) - \inf_{v \in \vv_\ld(P)} P (H_{x,v}) \right| < \e/3,
		\end{align*}
		by \autoref{propRHDcont_ld},
		and
		
		\item for any $n \geq N$, 
		\begin{align*}
			\sup_{x \in G} \left| \inf_{v \in \vv_{\ld-\dt}(P)} Q_n (H_{x,v}) - \inf_{v \in \vv_{\ld-\dt}(P)} P (H_{x,v}) \right| & <\e/3, \\
			\sup_{x \in G} \left| \inf_{v \in \vv_{\ld + \dt}(P)} Q_n (H_{x,v}) - \inf_{v \in \vv_{\ld + \dt}(P)} P (H_{x,v}) \right| & <\e/3,
		\end{align*}
		by \autoref{propRHDcpt}.
		
	\end{enumerate}
	The inequalities in \eqref{eq_thmRHDconsist} imply that, for all $x \in G$, 
	\begin{align*}
		D_\ld(x;P) - D_\ld(x;Q_n)
		& \leq D_\ld(x;P) - \inf_{v \in \vv_{\ld-\oa_n}(P)} Q_n(H_{x,v})
		\leq D_\ld(x;P) - D_{\ld-\oa_n}(x;P) + \e/3, \\
		D_\ld(x;Q_n) - D_\ld(x;P)
		& \leq \inf_{v \in \vv_{\ld+\oa_n}(P)} Q_n(H_{x,v}) - D_\ld(x;P)
		\leq \e/3 + D_{\ld+\oa_n}(x;P) - D_\ld(x;P).
	\end{align*}
	Combining these two bounds,
	we conclude that
	\begin{align*}
		& \sup_{x \in G} |D_\ld(x;Q_n) - D_\ld(x;P)|
		\\ \leq &  \e/3 + \sup_{x \in G} |D_\ld(x;P) - D_{\ld-\oa_n}(x;P)|
		+ \sup_{x \in G} |D_{\ld+\oa_n}(x;P) - D_\ld(x;P)|
		\\<& \e. 
	\end{align*}
\end{proof}

\begin{proof}[Proof of \autoref{thmRHD}]
	Since $|\Phi(a)-\Phi(b)| \leq (2\pi)^{-1/2} |a-b|$ for $a,b \in\R$, 
	$D_\ld(\cdot;P)$ is Lipschitz continuous with constant $(2\pi)^{-1/2} \ld^{-1}$;
	see \autoref{lemMlip}. 
	Thus, Conditions~\ref{condContPoint} and~\ref{condLimitCont} hold for the RHD.
	Condition~\ref{condContProb} follows by \autoref{thmRHDconsist}. 
\end{proof}

\subsubsection{Proofs of the propositions}

\begin{proof}[Proof of \autoref{propRHDcpt}]
	To use contradiction, we suppose that the convergence does not hold.
	Then, there exist $\e \in (0,\infty)$ and a sequence $\{x_n\}$ in $G$
	such that 
	\begin{align*}
		\left| \inf_{v \in \vv_\ld(P)} Q_n (H_{x_n,v}) - \inf_{v \in \vv_\ld(P)} P (H_{x_n,v}) \right| \geq \e, \quad \forall n \in \N. 
	\end{align*}
	We claim that there exists a sequence $\{v_n\}$ in $\vv_\ld(P)$ such that 
	\begin{align} \label{eq_propRHDcpt}
		|Q_n(H_{x_n,v_n}) - P(H_{x_n,v_n})| \geq 2\e/3.
	\end{align}
	To prove this claim, suppose that $\inf_{v \in \vv_\ld(P)} Q_n (H_{x_n,v}) - \inf_{v \in \vv_\ld(P)} P (H_{x_n,v}) \geq \e$ for all $ n \in \N $.
	Then, by definition of $\inf_{v \in \vv_\ld(P)} P (H_{x_n,v})$, 
	for each $n \in \N$,
	we can take $v_n \in \vv_\ld(P)$ such that 
	\begin{align*}
		P (H_{x_n,v}) \leq \inf_{v \in \vv_\ld(P)} P (H_{x_n,v}) + \e/3.
	\end{align*}
	This then implies 
	\begin{align*}
		Q_n(H_{x_n,v_n}) - P(H_{x_n,v_n})
		& \geq \inf_{v \in \vv_\ld(P)} Q_n (H_{x_n,v}) - \left\{ \inf_{v \in \vv_\ld(P)} P (H_{x_n,v}) + \e/3 \right\} 
		\\& \geq \e-\e/3 = 2\e/3.
	\end{align*}
	For the other direction, we suppose that $\inf_{v \in \vv_\ld(P)} P (H_{x_n,v}) - \inf_{v \in \vv_\ld(P)} Q_n (H_{x_n,v}) \geq \e$ for all $ n \in \N $.
	By replacing the roles of $P$ and $Q_n$, the same analysis yields $P(H_{x_n,v_n}) - Q(H_{x_n,v_n}) \geq 2\e/3$. 
	This verifies the claim.

	Since $G$ is compact, $\{x_n\}$ converges to some $x \in G$ following a subsequence. 
	Since the closed unit ball $\BB \equiv \{ v \in \HH:\|v\| \leq 1 \}$ is weakly compact by the Banach--Alaoglu theorem \citep[Theorem~3.15]{Rudin1991} and $\vv_\ld(P) \subseteq \BB$,
	$\{v_n\}$ weakly converges to some $v \in \BB$ following a subsequence.
%	Also, since $\vv_\ld(P)$ is weakly compact by \autoref{lemWeakCompact}, 
%	$\{v_n\}$ weakly converges to some $v \in \vv_\ld(P)$  following a subsequence.
%	Here, we prove that $v$ lies in $ \vv_\ld(P)$ using a similar argument to that for \autoref{lemWeakCompact}:
	Since $\ga(P)^{1/2}$ is compact and $\{v_n\}$ is bounded, $\ga(P)^{1/2}v_n$ has a strongly convergent subsequence \citep[Section~4.2]{mac09}, 
	whose limit must be the same as $v$,
	which implies $\|\ga(P)^{1/2}v\| = \lim_{n\to\infty} \|\ga(P)^{1/2}v_n\| \geq \ld$. 
	
	Using Skorohod's theorem \citep[Section~6]{bill99}, 
	we can find $Y_n \sim Q_n$ and $Y \sim P$ such that $Y_n \to Y$ almost surely on some probability space (following a subsequence).
	We can also verify that 
	\begin{align*}
		Z_n & \equiv \langle Y_n,v_n \rangle - \langle x_n, v_n \rangle 
		\to \langle Y, v \rangle - \langle x, v \rangle \equiv Z, \\
		\tilde{Z}_n & \equiv \langle Y,v_n \rangle - \langle x_n, v_n \rangle \to Z,
	\end{align*}
	almost surely. 
	Since $\pr(Z=0)=0$ due to the fact that $\var[\langle Y, v \rangle] = \|\ga(P)^{1/2}v\|^2 \geq \ld^2>0$ and Gaussianity of $Y$,
	%\tred{[check?? may just need continuity]} due to the fact that $v \in \vv_\ld(P)$ and Gaussianity, 
	continuous mapping theorem yields that the indicators converge as
	\begin{align*}
		\I(Z_n \geq 0) \to I(Z \geq 0), 
		\quad \I(\tilde{Z}_n \geq 0) \to \I(Z \geq 0),
	\end{align*}
	almost surely. 
	Finally, the dominated convergence theorem implies
	$Q_n(H_{x_n,v_n}) \to P(H_{x,v})$ and $P(H_{x_n,v_n}) \to P(H_{x,v})$,
	which further indicates $|Q_n(H_{x_n,v_n}) - P(H_{x_n,v_n})| \to 0$.
	This contradicts to our claim in \eqref{eq_propRHDcpt}, and the proof of the lemma is complete.
\end{proof}

\begin{proof}[Proof of \autoref{propRHDcont_ld}]
	Since $|\Phi(a)-\Phi(b)| \leq (2\pi)^{-1/2} |a-b|$ for $a,b \in\R$, 
	as $\beta\to\ld$, using \autoref{lemMcont_ld}, we have
	\begin{align*}
		\sup_{x \in G} | D_\beta(x;P) - D_\ld(x;P) | 
		& = \sup_{x \in G} | \Phi(-M_\beta(x-\eo[X])) = \Phi(-M_\ld(x-\eo[X])) |
		\\& \leq (2\pi)^{-1} \sup_{x \in G} |M_\beta(x-\eo[X]) - M_\ld(x-\eo[X])|
		\\& = (2\pi)^{-1} \sup_{f \in G-\eo[X]} |M_\beta(f) - M_\ld(f)|
		\to 0,
	\end{align*}
	where $G-\eo[X] \equiv \{ x-\eo[X] \in \HH: x \in G \}$.
\end{proof}

\subsubsection{Auxiliary lemmas}

%\begin{lem} \label{lemWeakCompact}
%	For any distribution $Q$ with a bounded linear covariance operator $\ga(Q)$, 
%	the direction set $\vv_\ld(Q)$ is weakly compact. 
%\end{lem}
%\begin{proof}
%	First, the Hilbert sphere $\{ v \in \HH:\|v\|=1 \}$ is weakly compact.
%	Since $\ga(Q)^{1/2}$ is bounded, it is continuous, 
%	yielding $v \mapsto \|\ga(Q)^{1/2}v\|$ is continuous and $\vv_\ld(Q)$ is closed.
%	Hence, $\vv_\ld(Q)$ is also weakly compact. 
%\end{proof}

For the remaining lemmas in this section, 
let $Q \in \pp(\HH)$ with positive-definite covariance operator $\ga(Q)$,
and write $\ga = \ga(Q)$ and $\vv_\ld  = \vv_\ld(Q)$.
We emphasize that the auxiliary lemmas below do not require Gaussianity.

\begin{lem} \label{lemSemiCont}
	The norm $\|\cdot\|$ in a Hilbert space $\HH$ is lower semicontinuous with respect to weak topology.
\end{lem}

\begin{proof}
	Let $\{x_n\}$ be a sequence in $\HH$ weakly converging to $x \in \HH$. Without loss of generality, we may assume that $x \neq 0$.
	Then, with $u=x/\|x\|$, Cauchy--Schwarz inequality implies
	\begin{align*}
		\|x\| = \langle x, u \rangle
		= \lim_{n\to\infty} \langle x_n, u \rangle
		\leq \liminf_{n\to\infty} \|x_n\|.
	\end{align*}
\end{proof}

\begin{lem} \label{lemBDapprox}
	Let $\ld \in (0,\|\ga^{1/2}\|)$ and $v \in \overline{\vv}_\ld = \{ v \in \HH:\|v\|=1, \|\ga^{1/2}v\|=\ld \}$. 
	Then, there exists a sequence $\{v_n\}$ in $\vv_\ld$ such that $\|\ga^{1/2}v_n\| > \ld$. 
\end{lem}

\begin{proof}
	Define $\ttt \equiv \ga - \ld^2 I$, where $I$ represents the identity operator on $\HH$. 
	Then, $\|\ttt^{1/2}v\|^2 = \|\ga^{1/2}v\|^2 - \ld^2 = 0$.
	
	We first consider the case where $\ttt v \neq 0$. 
	Define $u = \ttt v / \|\ttt v\|$ and $v_n \equiv (v + u/n) / \|v+u/n\|$ for each $n \in \N$. 
	Then, it is clear that $\|v_n\|=1$ and $v_n \to v$. 
	Also, for each $n \in \N$, it holds that
	\begin{align*}
		\|\ga^{1/2}v_n\|^2 - \ld^2
		& = \|\ttt^{1/2}v_n\|^2
		= {\|\ttt^{1/2}(v + u/n)\|^2 \over \|v+u/n\|^2}
		\\& = {\|\ttt^{1/2}v\|^2 + 2n^{-1} \langle \ttt v, u \rangle + n^{-2} \|\ttt^{1/2}u\|^2 \over \|v+u/n\|^2}.
	\end{align*}
	Here, since $\|\ttt^{1/2}v\|^2=0$ and $\langle \ttt v, u \rangle = \|\ttt v\|$ by the definition of $u$, 
	we have
	$\|\ga^{1/2}v_n\|^2- \ld^2 = (2n^{-1}\|\ttt v\| + n^{-2} \|\ttt^{1/2}u\|^2) / \|v+u/n\|^2 > 0$, i.e., $\|\ga^{1/2}v_n\|>\ld$. 
	
	We now suppose that $\ttt v = 0$ so that $\ld^2$ is an eigenvalue with eigenvector $v$. 
	since $\ld^2 < \|\ga\|$, there is a bigger eigenvalue $\rho > \ld^2$ with eigenvector $\phi$, where $\langle \phi, v \rangle = 0$.
	We now define $v_n \equiv (v+\phi/n ) / (1+n^{-2})^{1/2}$. 
	Clearly, $\|v_n\|=1$ and $v_n \to v$. 
	We also have that, for each $n \in \N$, 
	\begin{align*}
		\|\ga^{1/2}v_n\|^2
		= {\langle \ga (v+\phi/n), v + \phi/n \rangle \over 1+n^{-2}}
		= {\ld^2 + \rho/n^2 \over 1+n^{-2}} > \ld^2, 
	\end{align*}
	i.e., $\|\ga^{1/2}v_n\|>\ld$. 
	This completes the proof. 
\end{proof}

For each $\ld \in (0,\|\ga^{1/2}\|)$, % with $\vv_\ld  = \vv_\ld(P)$,
we write $M_\ld = M_\ld(\cdot;Q)$ for the function in Equation~\eqref{eq_RHDinMft}. 

\begin{lem} \label{lemRHDcont_ld}
	For each $f \in \HH$, $\ld \mapsto M_\ld(f)$ is continuous. 
\end{lem}

\begin{proof}
	Let $f \in \HH$ be fixed.
	We first note that the map $\ld \mapsto M_\ld(f)$ is non-increasing due to the nested property of the direction set $\vv_\ld$. 
	Therefore, this map has left and right limits. 
	We will investigate left and right continuity at $\ld \in (0,\|\ga^{1/2}\|)$ and prove that these match. 
	
	To see the left continuity, take a sequence $\{\beta_n\}$ of real numbers such that $\beta_n \uparrow \ld$ as $n\to\infty$. 
	By monotnicity, there exists 
	\begin{align} \label{eq_lemRHDcont_ld_LC1}
		L \equiv \lim_{\beta_n\uparrow\ld}M_{\beta_n}(f) \geq M_\ld(f)
	\end{align}
	To derive the other direction of the inequality, using definition of supremum,
	for each $n \in \N$, we find $v_n \in \vv_{\beta_n}$ such that
	\begin{align*}
		{|\langle f, v_n \rangle| \over \|\ga^{1/2}v_n\|}
		\geq M_{\beta_n}(f) - {1 \over n}. 
	\end{align*}
	Write $w_n = v_n/\|\ga^{1/2}v_n\|$ for each $n \in \N$ so that this new sequence $\{w_n\}$ is bounded as $\|w_n\| \leq \beta_n^{-1} \leq \beta_1^{-1}$ for all $n \in \N$. 
	Then, by the Banach--Alaoglu theorem \citep[Theorem~3.15]{Rudin1991},
	there is a weakly convergent limit $w \in \HH$ along a subsequence. 
	Regarding this limit $w$, we have the following observations.
	\begin{itemize}
		\item Since $\ga^{1/2}$ is compact, $1=\|\ga^{1/2}w_n\| \to \|\ga^{1/2}w\|$ \cite[cf.][Section~4.2]{mac09}, 
		meaning that $\|\ga^{1/2}w\| = 1$.
		\item Since the norm in a Hilbert space is lower semicontinuous with respect to weak topology (\autoref{lemSemiCont}),
		we also have
		\begin{align*}
			\|w\| \leq \liminf_{n\to\infty} \|w_n\| \leq \liminf_{n\to\infty} \beta_n^{-1} = \ld^{-1},
		\end{align*}
		yielding $w/\|w\| \in \vv_\ld$, because
		\begin{align*}
			\| \ga^{1/2} (w/\|w\|) \|
			= \|w\|^{-1} \geq \ld.
		\end{align*}
	\end{itemize}
	Taken together, we have
	\begin{align} \label{eq_lemRHDcont_ld_LC2}
		L 
		&= \lim_{\beta_n\uparrow\ld}M_{\beta_n}(f)
		\leq \lim_{n\to\infty} {|\langle f, v_n \rangle| \over \|\ga^{1/2}v_n\|}
		= |\langle f,w \rangle|
		= {|\langle f,w \rangle| \over\|\ga^{1/2}w\|}
		= {|\langle f,w/\|w\| \rangle| \over\|\ga^{1/2}(w/\|w\|)\|}
		\leq M_\ld(f).
	\end{align}
	The inequalities \eqref{eq_lemRHDcont_ld_LC1}--\eqref{eq_lemRHDcont_ld_LC2} conclude that the map $\ld \mapsto M_\ld(f)$ is left continuous.
	
	To see the right continuity, we take a sequence $\{\beta_n\}$ of real numbers such that $\beta_n \downarrow \ld$ as $n\to\infty$. 
	By monotnicity, there exists 
	\begin{align} \label{eq_lemRHDcont_ld_RC1}
		R \equiv \lim_{\beta_n\downarrow\ld}M_{\beta_n}(f) \leq M_\ld(f).
	\end{align}
	To prove $R \geq M_\ld(f)$, let $\e \in (0,\infty)$ be given. 
	Using definition of supremum, we take $v \in \vv_\ld$ such that $|\langle f, v \rangle|/\|\ga^{1/2}v\| > M_\ld(f) - \e$.
	If $\|\ga^{1/2}v\| > \ld$, then there exists $N \in \N$ such that for all $n \geq N$, $\beta_n < \|\ga^{1/2}v\|$, meaning that $v \in \vv_{\beta_n}$
	and further $M_{\beta_n}(f) \geq |\langle f, v \rangle|/\|\ga^{1/2}v\| > M_\ld(f) - \e$. 
	In this case, we thus have $L \geq M_\ld(f) - \e$. 
	We now suppose $\|\ga^{1/2}v\| = \ld$.
	Then, using \autoref{lemBDapprox}, we take a sequence $\{v_m\}$ such that $v_m \to v$ and $\|\ga^{1/2}v_m\| > \ld$ for all $m \in \N$. 
	Then, 
	there exists $N_1 \in \N$ such that for all $m \geq N_1$,
	\begin{align*}
		{|\langle f, v_m \rangle| \over \|\ga^{1/2}v_m\|}
		> {|\langle f, v \rangle| \over \|\ga^{1/2}v\|} - \e
		> M_\ld(f) - 2\e. 
	\end{align*}
	Let $m \geq N_1$. Then, there exists $N_2 = N_2(m) \in \N$ such that for all $n \geq N_2$, $\|\ga^{1/2}v_m\| > \beta_n$.
	Combining these,
	we have
	\begin{align*}
		M_{\beta_n}(f) 
		\geq {|\langle f, v_m \rangle| \over \|\ga^{1/2}v_m\|}
		> M_\ld(f)-2\e, \quad \forall n \geq N_2, 
	\end{align*}
	implying $R \geq M_\ld(f)-2\e$. 
	In brief, 
	\begin{align} \label{eq_lemRHDcont_ld_RC2}
		R \geq M_\ld(f)-2\e, \quad \forall \e \in (0,\infty).
	\end{align}
	The inequalities \eqref{eq_lemRHDcont_ld_RC1}--\eqref{eq_lemRHDcont_ld_RC2} conclude that the map $\ld \mapsto M_\ld(f)$ is right continuous.
\end{proof}

\begin{lem} \label{lemMlip}
	Let $[\ap, \beta) \in (0,\|\ga^{1/2}\|)$.
	Then, $M_\ld$ is Lipschitz continuous uniformly over $\ld \in [\ap, \beta)$ with common constant $\ap^{-1}$. 
\end{lem}

\begin{proof}
	First, for each $\ld \in (0,\|\ga^{1/2}\|)$,
	since $M_\ld(f) \leq \|f\|/\ld$ for all $f \in \HH$  and $M_\ld$ is a semi-norm on $\HH$, $M_\ld$ is Lipschitz continuous.
	More specifically, we derive
	\begin{align*}
		|M_\ld(f)-M_\ld(g)| \leq M_\ld(f-g) \leq \|f-g\|/\ld \leq \|f-g\|/\ap, \quad \forall f,g \in \HH. 
	\end{align*}
	This completes the proof. 
\end{proof}

\begin{lem} \label{lemMcont_ld}
	For a compact subset $G \subseteq \HH$, as $\beta\to\ld$, we have
	$\sup_{f \in G}|M_\beta(f)-M_\ld(f)|\to0$.
\end{lem}

\begin{proof}
	Let $\{\beta_n\}$ be a sequence such that $\beta_n\to\ld$ as $n\to\infty$. 
	Without loss of generality, we may assume that $\beta_n \in [\ap,\|\ga^{1/2}\|)$ for some $\ap \in (0,\ld)$. 
	Then, $\{M_{\beta_n}\}$ is equicontinuous by \autoref{lemMlip}. 
	By Arzel\'{a}--Ascoli theorem \cite[e.g.,][Section~10.1]{RF10}, the pointwise convergence of $\{M_{\beta_n}\}$ from \autoref{lemRHDcont_ld} implies the uniform convergence of $\{M_{\beta_n}\}$ over $G$.
\end{proof}

\subsection{Technical details for integrated and infimal depths} \label{appID}

\begin{enumerate}[(G)]
	\item \label{condGau}
	For each $u \in [0,1]$, 
	$P_u = \nd_d(\mu(u), \sa(u))$ with positive definite matrix $\sa(u)$ for all $u \in [0,1]$.
	%	the marginal distribution $P_u$ at $u$ is a Gaussian distribution on $\R^d$ with a positive definite covariance matrix.
\end{enumerate}
When Condition~\ref{condGau} holds,
throughout this appendix, we write $Z$ for a stochastic process  on $[0,1]$ 
defined by $Z(u) \equiv \sa(u)^{-1/2}\{ T(u)-\mu(u) \}$ for $u \in [0,1]$
so that $D_{\mathrm{HD}}(T(u), P_u) = \Phi( - \|Z(u)\| )$ for each $u \in [0,1]$.

\begin{enumerate}[({ITD}1)]
	
	\item \label{condITD1}
	There exist $\xi \sim \nd(0,1)$, a $C^d([0,1])$-valued Gaussian process $W$ independent of $\xi$, and a measurable map $\eta:[0,1]\to\R^d$ such that 
	$Z= W + \xi \eta$,
	where $T \sim P$. 
	
	\item \label{condITD3}
	For each $w \in C^d([0,1])$, 
	we define $F_w(a) \equiv \int_0^1 \Phi(- \|w(u)+a\eta(u)\| ) du$ for $a \in \R$. 
	Then, it holds that
	\begin{align*}
		\kll(\{a \in \R: F_w(a)=c\}) = 0,
		\quad \forall c \in \R, 
		\quad \forall w \in C^d([0,1]),
	\end{align*}
	%	almost every $c \in \R$ and $w \in C^d([0,1])$,
	where $\kll$ represents the Lebesgue measure on $\R$.

\end{enumerate}

Condition~\ref{condITD1} may be viewed as a one-dimensional Gaussian decomposition of the standardized process $Z$, 
while Condition~\ref{condITD3} is a non-flatness condition on the induced one-dimensional profile $F_w$. 
Together, these conditions reduce Condition~\ref{condLimitCont} for the ITD to the absence of atoms of $F_W(\xi)$, 
where $\xi \sim \nd(0,1)$ is independent of $W$. 
In this way, they isolate the nontrivial ingredient needed to exclude atoms of the ITD in a transparent one-dimensional form.

\begin{proof}[Proof of \autoref{thmITD}]
	To derive Condition~\ref{condContPoint}, 
	let $\{x_n\}$ be a sequence in $C^d([0,1])$ converging to $x \in C^d([0,1])$. 
	Then, for each $u \in [0,1]$, $HD(x_n(u), P_u) \to HD(x(u), P_u)$
	by \citet[Lemma~6.1]{DG92}.
	%	because $P_u$ is absolute continuous. 
	Then, by dominated convergence theorem, 
	Condition~\ref{condContPoint} holds as
	\begin{align*}
		D_{\mathrm{ITD}}(x_n;P)
		= \int_0^1 D_{\mathrm{multi}}(x_n(u), P_u) du
		\to \int_0^1 D_{\mathrm{multi}}(x(u), P_u) du
		= D_{\mathrm{ITD}}(x;P).
	\end{align*}
	Also, Condition~\ref{condContProb} is satisfied by \citet[Theorems~9 and~14]{NGOH16}.
	
	To prove the second assertion, 
	%	we further assume that $P_u = \nd_d(\mu(u), \sa(u))$ with positive definite matrix $\sa(u)$ for all $u \in [0,1]$.
	%	We then define a stochastic process $Z$ on $[0,1]$ as $Z(u) \equiv \sa(u)^{-1/2}\{ T(u)-\mu(u) \}$ for $u \in [0,1]$. 
	%	By Gaussianity,
	%	we have $D_{\mathrm{multi}}(x(u), P_u) = \Phi( - \|Z(u)\| )$
	%	and 
	note that $D_{\mathrm{ITD}}(T;P) = \int_0^1 \Phi(-\|Z(u)\|) du$.
	By Conditions~\ref{condITD1}, $D_{\mathrm{ITD}}(x;P) = F_W(\xi)$.
	Then, 
	since $\pr(F_w(\xi)=c) = \int_{a \in \R:F_w(a)=c} \varphi(a)da$,
	by Condition~\ref{condITD3},
	we observe that
	\begin{align*}
		\pr(D_{\mathrm{ITD}}(X;P) = c|W)
		= \pr(F_W(\xi)=c|W)
		= 0,
	\end{align*}
	implying $\pr(D_{\mathrm{ITD}}(X;P) = c)=0$. 
\end{proof}

\begin{enumerate}[({IFD})]
	\item \label{condIFD1}
	%	The process $T\sim P$ is Gaussian with continuous sample paths and its 
	The covariance process $\Sigma$ of $P$ is continuous with $\inf_{u \in [0,1]}\ld_{\min}(\Sigma(u)) > 0$,
	where $\ld_{\min}(A)$ denotes the minimum eigenvalue of a square matrix $A$.
\end{enumerate}

\begin{proof}[Proof of \autoref{thmIFD}]
	We	define a stochastic process $Z$ on $[0,1]$ as $Z(u) \equiv \sa(u)^{-1/2}\{ T(u)-\mu(u) \}$ for $u \in [0,1]$. 
	By Gaussianity,
	we have 
	%	$D_{\mathrm{multi}}(x(u), P_u) = \Phi( - \|\sa(u)^{-1/2}\{ x(u)-\mu(u) \}\| )$
	%	and 
	$D_{\mathrm{IFD}}(x;P) = \inf_{u \in [0,1]} \Phi(-\|\sa(u)^{-1/2}\{ x(u)-\mu(u) \}\|_{\R^d})$.
	Write $\|\cdot\|_{\R^{d\times d}}$ to represent a matrix norm. 
	Then, since $[0,1]$ is compact, $\sa$ is continuous, and $\sa(u)$ is positive definite for each $u \in [0,1]$, 
	we have $\sup_{u \in [0,1]} \|\sa(u)^{-1/2}\|_{\R^{d\times d}}<\infty$. 
	This implies that 
	\begin{align*}
		& |D_{\mathrm{IFD}}(x;P)-D_{\mathrm{IFD}}(y;P)|
		\\ \leq & \sup_{u \in [0,1]}|D_{\mathrm{multi}}(x(u), P_u)-D_{\mathrm{multi}}(y(u), P_u)|
		\\ \leq & \sup_{u \in [0,1]} \left| \Phi( - \|\sa(u)^{-1/2}\{ x(u)-\mu(u) \}\|_{\R^d} )
		- \Phi( - \|\sa(u)^{-1/2}\{ y(u)-\mu(u) \}\|_{\R^d} ) \right|
		\\ \leq & (2\pi)^{-1/2}\sup_{u \in [0,1]} \left| \|\sa(u)^{-1/2}\{ x(u)-\mu(u) \}\|_{\R^d} - \|\sa(u)^{-1/2}\{ y(u)-\mu(u) \}\|_{\R^d} \right|
		\\ \leq & (2\pi)^{-1/2} \sup_{u \in [0,1]} \|\sa(u)^{-1/2}\{ x(u)-y(u) \}\|_{\R^d}
		\leq M \|x-y\|,
	\end{align*}
	where $M \equiv (2\pi)^{-1/2} \sup_{u \in [0,1]} \|\sa(u)^{-1/2}\|_{\R^{d\times d}} $.
	This proves the first result. 
	Next, to derive Condition~\ref{condLimitCont} for the IFD,
	let $W(u,v) \equiv Z(u)^\top v$ for $u \in [0,1]$ and $v \in \SSS^{d-1} \equiv \{ v \in \R^d: \|v\|_{\R^d}=1 \}$.
	Then, by Conditions~\ref{condGau} and~\ref{condIFD1}, this new process $W$ is a centered Gaussian process with continuous sample paths on the compact set $[0,1]\times \SSS^{d-1}$,
	and $\var[W(u,v)] = 1$ for all $u,v$. 
	Then, by \citet[Theorem~1]{Yang25}, $\|Z\|_{\sup} =  \sup_{u \in [0,1],v \in \SSS^{d-1}}|W(u,v)|$ is a continuous random variable, and hence, so is $D_{\mathrm{IFD}}(T;P) = \Phi(-\|Z\|_{\sup})$.
	Finally, Condition~\ref{condContProb} follows by \citet[Theorem~7]{GN15}.
\end{proof}

\subsection{Assumptions for FoFR residual bootstrap consistency} \label{AppFoFRcond}

We provide the technical conditions that are sufficient for the residual bootstrap consistency for function-on-function regression in \autoref{thm_fofr}; we refer to \cite{Yeon26} for more details about its theoretical development and an example that satisfies all these conditions (Example~1 therein).
We write $\ga \equiv \eo[(X-\eo[X])^{\otimes 2}]$ for the covariance operator of $X$ with its spectral decomposition $\ga = \sum_{j=1}^\infty \g_j (\phi_j \otimes \phi_j)$, 
where $(\g_j,\phi_j)$ denotes its $j$-th eigenvalue-eigenfunction pair.
\autoref{thm_fofr} holds under the following conditions:

\begin{enumerate}[(A0)]
	\item $\ker \ga = \{0\}$; \label{condModel}
\end{enumerate}
\begin{enumerate}[({A}1)]
	\item the slope operator $\BB$ is a Hilbert--Schmidt operator on $\HH$ such that the sequence 
	\begin{align*}
		\left\{ \sum_{j'=1}^\infty \langle \BB \phi_j, \phi_{j'} \rangle^2 \right\}_{j=1}^\infty
	\end{align*}
	is non-increasing; \label{condSlope}
	
	\item $\sup_{j \in \N} \g_j^{-2} \eo[ \langle X_1-\eo[X_1], \phi_j \rangle^4 ] < \infty$; \label{condMoment}
	
	\item the eigenvaules $\{\g_j\}_{j=1}^\infty$ satisfy $\g_j = \varphi(j)$ for sufficiently large $j$, where $\varphi:(0,\infty) \to (0,\infty)$ is a convex function; \label{condEVconvex}
	
	\item $\sup_{j \in \N} \g_j j \log j <\infty$; \label{condEVdecay}
	
	\item $J^{-1} + n^{-1/2} J^2 \log J + n^{-1} \sum_{j=1}^J \dt_j^{-2} \to 0$ as $n\to\infty$, where $\{\dt_j\}$ are eigengaps defined as $\dt_j \equiv \g_1 - \g_2$ if $j=1$ and $\dt_j \equiv \min \{\g_j - \g_{j+1}, \g_{j-1} - \g_j \}$ if $j > 1$ otherwise;
	\label{condBasicTrunc}
	
\end{enumerate}

\begin{enumerate}[(B)]
	\item $\sup_{j \in \N} j^{-1} m(j,u) \|B\phi_j\|^2 < \infty$ for $u>5$, \label{condBias}
\end{enumerate}
where
\begin{align}
	m(j,u) \equiv \max \left\{ j^u, \sum_{l=1}^j \dt_l^{-2} \right\}, \quad j \in \N, \quad u >0; \label{eq_mju}
\end{align}

\begin{enumerate}[(S)]
	\item $J \tau_J(X_0)^{-1} = O_\pr(1)$, where $\tau_J(x) \equiv \sum_{j=1}^J \g_j^{-1} \langle x-\eo[X_1], \phi_j \rangle^2$ for $x \in \HH$; \label{condScale}
\end{enumerate}

\begin{enumerate}[(E)]
	\item there exists a class $\{g_x\}_{x \in \HH}$ of (non-random) functions on $[0,\infty)$ to $[0,\infty)$ with $\lim_{u\to\infty}g_x(u) = 0$ for each $x \in \HH$ such that  $\eo[\langle \e, x \rangle^2 \I(|\langle \e, x \rangle| \geq u)|X] \leq g_x(u)$ almost surely for any $u \geq 0$ and $x \in \HH$; \label{condError}
\end{enumerate}

\begin{enumerate}[(R)]
	\item $\lim_{n\to\infty} J/J_{\mathrm{cen}} \geq 1$; \label{condRatio}
\end{enumerate}

\begin{enumerate}[(O)]
	\item the truncation levels $J_{\mathrm{res}}$ and $J_{\mathrm{cen}}$ satisfy Condition~\ref{condBasicTrunc} and Condition~\ref{condScale} holds for $J_{\mathrm{cen}}$; and \label{condOther}
\end{enumerate}

\begin{enumerate}[(T)]
	\item $n^{-1/2}J^{5/2} (\log J)^J \to 0$ and $n = O(m(J,u))$ for $u>5$. \label{condTrunc}
\end{enumerate}

\section{Some details in numerical studies} \label{AppNumDetails}

\subsection{Implementation of the RHD} \label{AppRHDimple}

The random projection approach is adopted to approximate the sample RHD. 
%In terms of the modified RHD \eqref{eqRHDmod},
Namely, we generate directions $\{v_m\}_{m=1}^M$ from $\{v \in \HH:\|v\|=1\}$, choose the $u$-th quantile of $\{ \|\widehat{\ga}^{1/2}v_m\| \}_{m=1}^M$ to use the regularization $\ld$, and compute the minimum of the empirical halfspace probabilities (cf.~\autoref{AppRHDimple} of the supplement).
Here, $\widehat{\ga} \equiv \ga(\widehat{Q}_n) = n^{-1} \sum_{i=1}^n (X_i-\bar{X}) ^{\otimes 2}$ denotes the sample covariance operator of the functional observations $\{X_i\}_{i=1}^n$, where $\bar{X} \equiv n^{-1} \sum_{i=1}^nX_i$ and $\widehat{Q}_n$ denotes the empirical distribution of $\{X_i\}_{i=1}^n$.
The quantile level $u$ may provide different perspectives for statistical analysis;
with smaller $u$, the resulting RHD considers more projection directions whose variances get smaller, and hence, becomes better sensitive to shapes of functional observations.
The RHD then leads to developing an effective method to detect both magnitude and shape functional outliers.

To describe its implementation, 
recall that the sample counterpart of the RHD of $x \in \HH$ in Equation~\eqref{eqRHDmod} of \autoref{ssecFdepthRHD} is given as
\begin{align*}
	D_{\mathrm{RHD},\ld}(x;\widehat{Q}_n) 
	= \inf_{v \in \vv_\ld(\widehat{Q}_n)} n^{-1} \sum_{i=1}^n \I( \langle X_i - x, v \rangle \geq 0).
\end{align*}
We adopt a random projection approach to approximate the sample RHD $D_{\mathrm{RHD},\ld}(x;\widehat{Q}_n)$. 
We first draw $M$ random directions $\{\tilde{v}_m\}_{m=1}^M$ from $\vv_\ld(\widehat{Q}_n)$
and then compute the minimum instead of the infimum:
\begin{align*}
	\tilde{D}_{\mathrm{RHD},\ld}(x;\widehat{Q}_n)
	\equiv \min_{1 \leq m \leq M} n^{-1} \sum_{i=1}^n \I( \langle X_i - x, v_m \rangle \geq 0). 
\end{align*}

To explain our tie-breaking approach, 
we first find a direction that shows the largest outlyingness.
Namely, 
for each $x$, we may find $L = L(x)$ minimizers as
\begin{align*}
	\{\tilde{w}_1(x), \dots, \tilde{w}_L(x)\}
	= \argmin_{\tilde{v} \in \{\tilde{v}_m\}_{m=1}^M}n^{-1} \sum_{i=1}^n \I( \langle X_i - x, v_m \rangle \geq 0).
\end{align*}
and take average to obtain $\bar{w}(x) \equiv L^{-1} \sum_{l=1}^L \tilde{w}_l(x)$. 
We now consider the outlyingness $O_{\bar{w}(x)} (x;\widehat{Q}_n)$ of the projection of $x$ onto $\bar{w}(x)$,
where 
\begin{align*}
	O_v(x;P)
	\equiv {|\langle x, v \rangle - \mathsf{med}[\langle X, v \rangle]| \over \mathsf{MAD}[\langle X, v \rangle]}
\end{align*}
with median $\mathsf{med}$ and median absolute deviation  $\mathsf{MAD}$ operators.

Suppose we have several $x_1, \dots, x_Q$ that are tied in RHD rankings.
Then, we break these ties using the outlyingness in the projection directions $\{\bar{w}(x_q)\}_{q=1}^Q$ that minimize the halfspace probabilites.
That is, the rankings of $\{O_{\bar{w}(x_q)} (x_q;\widehat{Q}_n)\}_{q=1}^Q$ are assigned to the rankings of $\{x_q\}_{q=1}^Q$, 
where lower rankings mean higher outlyingness $O_v(x)$.

%define sample version, 
%
%define approximate,
%
%09/09/2025: break ties in RHD rankings
%
%Based on Algorithms in \cite{YDL25RHD}.
%
%Recall that the approximate RHD is 
%\begin{align*}
%	\tilde{D}_{\ld,n,M}(x) = \min_{\hat{\bm{a}} \in \tilde{\aaa}_{\ld,J,M}} n^{-1} \sum_{i=1}^n \I( (\hat{\bm{X}}_i - \hat{\bm{x}})^\top \hat{\bm{a}} \geq 0 ).
%\end{align*}

\subsection{Orthogonal systems} \label{AppOrtho}

We provide the detailed constructions of the orthonormal systems used for the numerical studies in \autoref{ssec_5_1},
which are also adopted by \cite{YK26}. 
The trigonometric functions $\{\phi_{\mathrm{tri},j}\}_{j=1}^{J_{\mathrm{true}}}$ are defined as 
\begin{align*}
	\phi_{\mathrm{tri},1}(t) = 1,
	\phi_{\mathrm{tri},2m}(t) = \sqrt{2} \sin (2m\pi t),
	\phi_{\mathrm{tri},2m+1}(t) = \sqrt{2} \cos (2m\pi t),
	m \in \N.
\end{align*}
For the remaining three orthornormal systems, we use the following three sets $\{f_{q,j}\}_{j=1}^{J_{\rm true}}$ for $q \in \{\rm mono, cheb, spl\}$ of linearly independent functions.

\smallskip

\noindent{\it Monomials:} \
\[
f_{\mathrm{mono},j}(t) = t^j,
\quad t \in [0,1],
\quad j=1,\dots,J_{\rm true}.
\]

\smallskip

\noindent{\it Shifted  Chebyshev polynomials:} \
The (unshifted) Chebyshev polynomials
$\{f_{\mathrm{cheb},j}\}_{j=1}^{J_{\mathrm{true}}}$
of the first kind are determined by the recurrence
formula $f_{\mathrm{cheb,unshift},j}(s) =
2sf_{\mathrm{cheb,unshift},j-1}(s) + f_{\mathrm{cheb,unshift},j-2}(s)$ for $j \geq 2$
with $f_{\mathrm{cheb,unshift},1}(s)=s$ and $f_{\mathrm{cheb,unshift},0}(s)=1$
where $s \in [-1,1]$ \cite[Chapter~3]{Caro98};
these are implemented by \texttt{chebPoly} function in \texttt{pracma}
R package. We shift them to $[0, 1]$
\[
f_{\mathrm{cheb},j}(t) = f_{\mathrm{cheb,unshift},j}(2t-1), \quad t \in [0,1], \quad j=1,\dots,J_{\mathrm{true}}.
\]

\smallskip

\noindent{\it B-spline basis functions of order 4:} \
We use \texttt{create.bspline.basis} function in \texttt{fda} package to create well-known B-spline basis functions of order 4 (i.e., cubic splines) and denote them by $\{f_{\mathrm{spl},j}\}_{j=1}^{J_{\rm true}}$.

\smallskip

\noindent
For each $q \in \{\rm mono, cheb, spl\}$,
we orthogonalize the set $\{f_{q,j}\}_{j=1}^{J_{\rm true}}$
by using the Gram--Schmidt orthogonalization \cite[Theorem~2.4.10]{HE15},
then normalize them to unit norm,
and the resulting orthonormal system is denoted as $\{\phi_{q,j}\}_{j=1}^{J_{\mathrm{true}}}$.

\subsection{Construction of alternative mean functions} \label{AppAlter}

In this section, we provide the alternative mean functions considered for the numerical studies in \autoref{ssec_5_1}, which are plotted in \autoref{fig_2test_alter}.

\begin{itemize}
	\item $\mu_{2,\mathrm{Mag}}(t) = 1$

	\item $\mu_{2,\mathrm{Jump}}(t) = -2\I(t \leq 0.2)+1$
	
	\item $\mu_{2,\mathrm{Peak}}(t) = -2\I(0.2<t\leq0.4)+1$

	\item $\mu_{2,\mathrm{Lin}}(t) = 2t-1$

	\item $\mu_{2,\mathrm{Quad}}(t) = 8(t-1/2)^2-1$
	
	\item $\mu_{2,\mathrm{Cub}}(t) = 12 \sqrt{3}t(t-1/2)(t-1)$
	
	\item $\mu_{2,\mathrm{Wig}}(t) = \sin(10\pi(t-0.05))$
	
\end{itemize}

\subsection{Implementation of the classical summary statistics for two-sample inference} \label{AppFVE2test}

In the two-sample framework in \autoref{ssec_2test},
let $\widehat{\ga}_{\mathrm{pool}}$ be the pooled covariance operator,
i.e.,
\begin{align*}
	\widehat{\ga}_{\mathrm{pool}}
	= {1 \over n_1+n_2} \sum_{k=1}^2 \sum_{i=1}^{n_k} (X_{ki} - \bar{X}_k)^{\otimes 2}. 
\end{align*}
Write $\hat{\g}_{\mathrm{pool},j}$ for the $j$-th eigenvalue of $\widehat{\ga}_{\mathrm{pool}}$,
and consider the fraction of variance explained at the truncation $J$ as
\begin{align*}
	FVE(J) \equiv {\sum_{j=1}^J \hat{\g}_{\mathrm{pool},j} \over \sum_{j=1}^{n_1+n_2} \hat{\g}_{\mathrm{pool},j}}.
\end{align*}
Then, for a threshold $\rho \in (0,1)$, we pick the smallest $\widehat{J}$ such that $FVE(\widehat{J})$ exceed $\rho$,
that is,
\begin{align*}
	\widehat{J} \equiv \argmin\{ J: FVE(J) \geq \rho \}. 
\end{align*}
This selection strategy is used for the methods QuadInv and QuadMul explained in \autoref{ssec_5_1}.

\subsection{Implementation of the FoFR estimation} \label{AppDetailFoFR}

This section describes the truncation level selection for the residual bootstrap used in \autoref{ssec_5_2}.

Consider a set $J_{\mathrm{res}}$ of candidate truncation levels, e.g., $\jj_{\mathrm{res}} = \{1, \dots, 20\}$.
We assume $n = Gm$ for simplicity, e.g., $G = 5$, and divide the original sample $\dd \equiv \{ (Y_i, X_i) \}_{i=1}^n$ into $G$ groups, say $\dd_1, \dots, \dd_G$. 
For each $J_{\mathrm{res}} \in \jj_{\mathrm{res}}$ and $g = 1, \dots, G$, we perform the following.
\begin{enumerate}
	\item 
	Use the $g$-th dataset $\dd_g$ as the testing sample and the other datasets $\bigcup_{g' \neq g} \dd_g$ as the training samples.
	The sample sizes are $m$ and $n-m$, and we write $\ii_{\mathrm{test}}$ and  $\ii_{\mathrm{tr}}$ for the corresponding index sets.

	\item 
	Using the training sample $\{ (Y_i, X_i) \}_{i \in \ii_{\mathrm{tr}}}$, compute the estimator $\widehat{\BB}_{\mathrm{tr},J_{\mathrm{res}}}$ with truncation $J_{\mathrm{res}}$. 
	
	\item 
	Using the testing sample $\{ (Y_i, X_i) \}_{i \in \ii_{\mathrm{test}}}$,
	compute the prediction error
	\begin{align*}
		PE_g(J_{\mathrm{res}}) \equiv m^{-1} \sum_{i \in \ii_{\mathrm{test}}} \|Y_i - \bar{Y} - \widehat{\BB}_{\mathrm{tr},J_{\mathrm{res}}} (X_i - \bar{X})\|^2.
	\end{align*}

\end{enumerate}
We then select an optimal truncation level $\widehat{K}$ that minimized the averaged estimated prediction error, i.e.,
\begin{align*}
	\widehat{J}_{\mathrm{res}}
	= \argmin_{J_{\mathrm{res}} \in \jj_{\mathrm{res}}} G^{-1} \sum_{g=1}^G PE_g(J_{\mathrm{res}}).
\end{align*}

Upon setting $\widehat{J}_{\mathrm{cen}} = \widehat{J}_{\mathrm{res}}$, the main truncation level $J$ is chosen by the fraction of variance explained similarly that in \autoref{AppFVE2test}.
We define the fraction of variance explained at the truncation $J$ as
\begin{align*}
	FVE(J) \equiv {\sum_{j=1}^J \hat{\g}_j \over \sum_{j=1}^n \hat{\g}_j}.
\end{align*}
Then, for a threshold $\rho \in (0,1)$, we pick the smallest $\widehat{J}$ (but no less than $\widehat{J}_{\mathrm{cen}}$) such that $FVE(\widehat{J})$ exceed $\rho$,
that is,
\begin{align*}
	\widehat{J} \equiv \argmin\{ J \geq \widehat{J}_{\mathrm{cen}}: FVE(J) \geq \rho \}. 
\end{align*}

\end{document}